%% file: main.tex
\documentclass[sigconf,nonacm=True,dvipsnames]{acmart}
\usepackage{color}

\newcommand{\Hy} {\mathscr H}
\newcommand{\Gr} {\mathscr G}

\usepackage{bm}
\usepackage{inputenc}
\usepackage{amsmath}
\usepackage{subcaption}
\usepackage{cleveref}

\usepackage{mathtools}
\setcopyright{none} 
\usepackage{comment}
\usepackage{appendix}
\usepackage{bbold}
\mathchardef\UrlBreakPenalty=32767

\begin{document}

\title[Consensus dynamics on temporal hypergraphs]{Consensus dynamics on temporal hypergraphs}

\author{Leonie Neuhäuser}
\authornotemark[1]
\email{neuhaeuser@cs.rwth-aachen.de}
\affiliation{%
  \institution{RWTH Aachen University}
  \country{Germany}
}
\author{Renaud Lambiotte}
\authornotemark[2]
\email{lambiotte@maths.ox.ac.uk}
\affiliation{%
  \institution{University of Oxford}
  \country{UK}
}
\author{Michael T. Schaub}
\authornotemark[3]
\email{schaub@cs.rwth-aachen.de}
\affiliation{%
  \institution{RWTH Aachen University}
  \country{Germany}
}

\renewcommand{\shortauthors}{Neuhäuser, Lambiotte, Schaub}

\begin{abstract}
We investigate consensus dynamics on temporal hypergraphs that encode network systems with time-dependent, multi-way interactions. 
We compare this dynamics with that on appropriate projections of this higher-order network representation that flatten the temporal, the multi-way component, or both. 
For linear average consensus dynamics, we find that the convergence of a randomly switching time-varying system with multi-way interactions is slower than the convergence of the corresponding system with pairwise interactions, which in turn exhibits a slower convergence rate than a consensus dynamics on the corresponding static network.
We then consider a nonlinear consensus dynamics model in the temporal setting.
Here we find that in addition to an effect on the convergence speed, the final consensus value of the temporal system can differ strongly from the consensus on the aggregated, static hypergraph. 
In particular we observe a first-mover advantage in the consensus formation process: If there is a local majority opinion in the hyperedges that are active early on, the majority in these first-mover groups has a higher influence on the final consensus value --- a behaviour that is not observable in this form in projections of the temporal hypergraph. 

\end{abstract}

\def\UrlFont{\rmfamily\small}
\mathchardef\UrlBreakPenalty=10000

\maketitle
\input{0_introduction}
\input{1_preliminaries}
\input{2_model}
\input{3_Linear_temporal_effects}

\input{4_Nonlinear_temporal_effects}

\input{5_conclusion}

\section*{acknowledgements}
Michael T. Schaub and Leonie Neuhäuser acknowledge funding by the Ministry of Culture and Science (MKW) of the German State of North Rhine-Westphalia (“NRW Rückkehrprogramm”). Renaud Lambiotte acknowledges support from the EPSRC Grants No.EP/V013068/1 and EP/V03474X/1.

\urlstyle{same}

\bibliographystyle{ACM-Reference-Format}
\bibliography{bib}

\newpage
\input{6_appendix}

\end{document}

%% file: 0_introduction.tex
\section{Introduction}
\label{sec:introduction}
In a number of application scenarios, including the study of epidemic spreading and opinion formation in social networks, we are interested in understanding a dynamical process taking place over a time-varying network structure. Such system can often be modelled as time-switched system of the form
$\dot{x}(t) = M^{(\sigma (t))}x(t)$, where $\sigma(t): \mathbb{R}\rightarrow \mathbb{N}$ is a switching signal which selects an interaction matrix $M^{(i)}\in \mathcal{M}$ from a set of matrices $\mathcal M$ for any point in time.

Thus the sequence of operators $M^{(0)}, \cdots, M^{(r-1)}$ determines together with the switching signal $\sigma(t)$ the evolution of the node state vector $x(t)\in \mathbb{R}^n$ over time.
Here, the form of the operator $M^{(i)}$ is determined both by the \emph{type of dynamics} (e.g. diffusion, consensus, synchronisation, etc.) and the \emph{topology} of the possible interactions.

Concrete examples of dynamical processes that may be modelled within this framework include communication patterns from phone calls~\cite{karsai_correlated_2012} or collective behaviour of animals~\cite{zschaler_adaptive-network_2012} and human navigation~\cite{singer_detecting_2014}. 
In this context, the sequence $M^{(0)}, \cdots, M^{(r-1)}$ of time-stamped interactions are typically referred to as temporal networks and encode pairwise interactions between nodes.
A large body of research has focused on elucidating how the temporal switching of the topology may impact the resulting dynamics compared to the scenario in which all the appearing edges are present simultaneously (in a so-called aggregated network model).
Notably, the temporally limited presence and order of interactions between nodes in a temporal network can create dependencies in the interaction paths of the nodes which are not described well by the aggregated network.
As a result, some nodes may be able to interact in the aggregate network, while they cannot interact in the temporal network. 
Dynamical processes evolving over the aggregated network can thus heavily differ from the dynamics supported on the temporal topology. 
For networks, it has been shown that e.g. diffusion can either speed up or slow down~\cite{masuda_temporal_2013, scholtes_causality-driven_2014} depending on the structural and temporal properties of the network.

This former research has broadly focused on the case where the matrices $M^{(i)}$ encode interactions takings place over a graph~\cite{masuda_temporal_2013}. 
In this setup, the interactions between two nodes are assumed to be pairwise which can be a limiting assumption, in general.
Recently there have been a number of works on extending graph-based models to multi-way interaction frameworks such as hypergraphs or simplicial complexes account for group interactions \cite{schaub_random_2018,mukherjee_random_2016,parzanchevski_simplicial_2017,muhammad_control_nodate,petri_simplicial_2018, masulli_algebro-topological_nodate,carletti_random_2020,neuhauser_multibody_2020}.  
These can appear in different situations ranging from social and collaborative settings, where people interact in groups rather in pairs~\cite{patania_shape_2017} to joint neuronal activity in brains \cite{petri_homological_2014}. 
Specifically, for dynamics on hypergraphs, it was shown that higher-order interactions can significantly modify the overall dynamical process in comparison to the corresponding dynamical process supported on a graph.
However, the focus of all these works was on multi-way interactions with a (temporally) fixed topology.

Combining the above two research strands, in this work we are interested in time-switched dynamics evolving over temporal hypergraphs. 
The issue of temporal ordering of interactions, which is at the core of the study of temporal networks, is also present in general interaction topologies described by hypergraphs, as the example in \Cref{aggreagetemporal} shows.
Thus (higher-order) dynamical effects may arise both due to a temporal ordering of the interactions and due to the presence of multi-way interactions.  It is however an open problem to look at the interaction effects of these temporal and topological dimensions.
\begin{figure}
\centering
\includegraphics[width=\columnwidth]{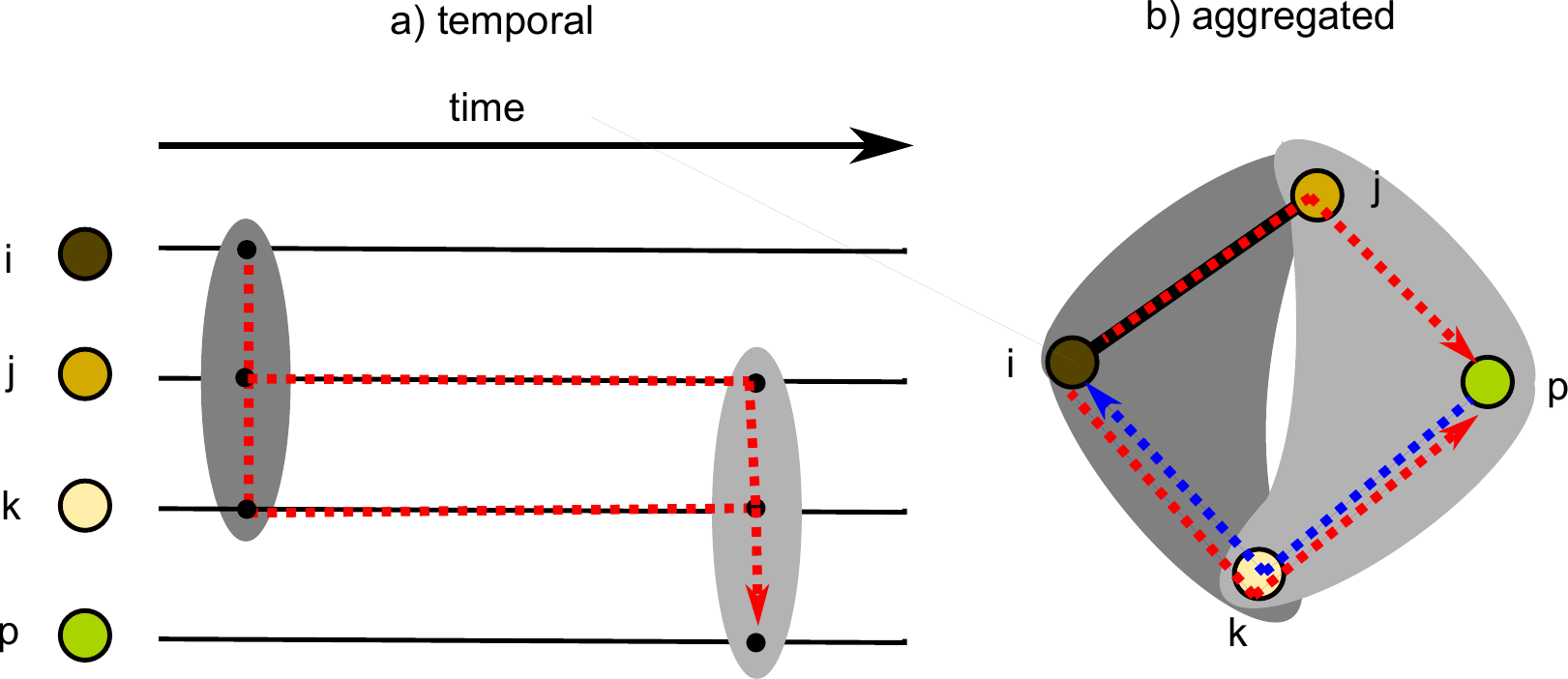}
\caption[Temporal and aggregate topology]{In the temporal topology, only the indirect interaction $i \rightarrow p$ is possible via two possible temporal paths (visualised by the red arrows), but the backward flow of information $p \rightarrow i$ (blue arrow) not due to the temporal order of the interactions. Nevertheless, both indirect interactions are present in the aggregate topology. }
\label{aggreagetemporal}%
\end{figure}
In this paper, we thus investigate the combined effect of these two types of higher-order for networks, their temporality and their multi-way interactions. In each case, some form of projection can be used to flatten the higher-order representation. In the case of temporal networks, it is standard to aggregate different time windows to build a static, often weighted, network representation. In the case of hypergraphs, different types of reductions, e.g. the clique reduced multi-graph, can be used to transform the multi-way representation into a two-way, also often weighted network model. Here, we consider temporal hypergraphs, for which two types of projection can thus be constructed. We will use the term aggregate when referring to a projection in the temporal dimension, and reduced when referring to a projection in the multi-way dimension. When assessing the impact of higher-order, we will compare the dynamics on temporal hypergraphs and these different types of projections.
\paragraph{Contributions}
In particular, we extend the line of higher-order model research by investigating continuous dynamics on temporal hypergraphs and thus examine the interaction effect of two types of higher-order facets.
First, we consider linear dynamics on temporal hypergraphs, more specifically linear consensus dynamics. 
For this, we extend the random, time-switched model in \cite{masuda_temporal_2013} using the fact that linear dynamics on hypergraphs can always be rewritten as an equivalent pairwise dynamics on an appropriately weighted network \cite{neuhauser_multibody_2020,neuhauser2021consensus}. 
We observe an effect on the convergence speed of the system: consensus dynamics with 3-way (time-varying) interactions are slowed down in comparison to (time-varying) pairwise interactions.
This happens in addition to the slow-down that can be observed when considering consensus dynamics on temporal vs aggregated, static networks. 

Second, we go beyond linear interactions and investigate nonlinear consensus dynamics on temporal hypergraphs. 
In particular, we extend the nonlinear consensus model introduced in \cite{neuhauser_multibody_2020} to temporal settings. 
We observe, that there is not only a temporal effect on the convergence speed but also on the consensus value of the system, which differs from the consensus on the aggregated, static hypergraph. 
More specifically, we can observe a first-mover advantage.
If there is a local majority in the hyperedges which are active early on, the majority in these first-mover groups has a higher influence on the final consensus value --- a behaviour that is not observable in this form in projections of the temporal hypergraph. 

\paragraph{Outline}
The rest of the paper is structured as follows. 
In \Cref{sec:preliminaries}, we first introduce the notation and give an overview about related work. In \Cref{sec:model} we then introduce the general consensus model for static hypergraphs that we will extend to temporal settings and the time-switched system model that we use for the analysis of the linear consensus dynamics on temporal hypergraphs and state general results. Building on these results, in \Cref{sec:linear} we discuss the influence of temporal ordering in the case of linear dynamics on hypergraphs. We then go beyond linear dynamics and focus on nonlinear consensus dynamics on hypergraphs in \Cref{sec:nonlinear}. 
Finally, we discuss our findings and future research in \Cref{sec:conclusion}.

%% file: 1_preliminaries.tex
\section{Preliminaries}
\label{sec:preliminaries}
\subsection{Notation}
Let $\Gr$ be an undirected graph consisting of a set $V(\Gr)=\{1,\dots,N\}$ of $N$ nodes connected by a set of edges $E(\Gr)=\{\{i,j\}: i,j \in V(\Gr)\}$, described by unordered tuples of nodes. 
The structure of the network can  be  represented by the adjacency matrix $A \in  \mathbb{R}^{N \times N}$ with entries
\begin{align}
	A_{ij}=\begin{cases}1 & (i,j) \in E(\Gr ) \\
		0 & \text{ otherwise}.\end{cases}
\end{align}
In the case of undirected networks, the adjacency matrix $A$ is symmetric.
To encode possible multi-way interactions in a dynamical system we use a hypergraph $\Hy$.
A hypergraph $\Hy$ consist of a set $V(\Hy) = \{ 1, 2, \hdots, N \}$ of $N$ nodes, and a set $E(\Hy) = \{ E_1, E_2, \hdots, E_K\}$ of $K$ hyperedges. 
Each hyperedge $E_\alpha$ is a subset of the nodes, i.e. $E_\alpha \subseteq V(\Hy)$ for all $\alpha =1,2,\hdots,K$, where each hyperedge may have a different cardinality $|E_\alpha|$. 
A graph is thus simply a hypergraph constrained to contain only $2$-edges. In this work, we will focus on comparing graphs and hypergraphs with hyperedges of cardinality $|E_\alpha|=3$. 
The structure of a hypergraph can  be  represented by an adjacency tensor, in the case of a 3-hypergraph it is given by $\mathbf{A} \in  \mathbb{R}^{N \times N \times N}$ with entries
\begin{align}
	\mathbf{A}_{ijk}=\begin{cases}1 & \{i,j,k\} \in E(\Hy ) \\
		0 & \text{ otherwise}.\end{cases}
\end{align}

\subsection{Related work}
\subsubsection{Dynamics on temporal networks}
Several works have investigated the influence of temporality on dynamics on networks. As a result, it has been shown that temporal orderings can either speed up or slow down diffusion \cite{masuda_temporal_2013, scholtes_causality-driven_2014}, depending on the particular structure of the network. 
Therefore, higher-order models \cite{rosvall_memory_2014, scholtes_causality-driven_2014} have been developed to extend static (or aggregated) network models in order to account for temporal effects. 

The work of Masuda et. al. \cite{masuda_temporal_2013} is particularly relevant for our setting. 
The authors investigated general interaction sequences to examine temporal effects on linear diffusion dynamics mediated by pairwise interactions. 
The interaction matrix $M^{(i)}$ in their case thus becomes the negative graph Laplacian of a single edge-coupling.  
The main result of \cite{masuda_temporal_2013} is that diffusive dynamics are generally slower in temporal networks.
In the first part of this work we generalise these results from diffusive edge couplings to general interaction matrices with certain algebraic properties in order to investigate (linear) consensus dynamics on hypergraphs. 

\subsubsection{Dynamics on hypergraphs}
Recent works have focused on investigating dynamics on higher-order topologies~\cite{carletti_dynamical_2020}, and a variety of dynamical processes such as diffusion~\cite{schaub_random_2018,carletti_random_2020}, synchronization~\cite{skardal_abrupt_2019} or social and opinion dynamics~\cite{neuhauser_multibody_2020,sahasrabuddhe_modelling_2020,iacopini_simplicial_2019} have been extended to multi-body frameworks. 
As a result, it was shown that higher-order interactions can significantly modify the dynamical process in comparison to the dynamical process on the underlying reduced network. 

In previous work~\cite{neuhauser_multibody_2020,neuhauser2020opinion,sahasrabuddhe_modelling_2020,neuhauser2021consensus} we showed that linear dynamics on hypergraphs can always be rewritten in terms of the underlying network, when properly rescaled. 
Nonlinear interactions are thus necessary in order to reveal dynamical effects that are different from the reduction of the dynamics to a network setting.

\subsubsection{Research on temporal hypergraphs}
Recent work which is extending temporal networks to multi-way interaction frameworks such as hypergraphs has mostly focused on analysing the evolution and formation of groups in empirical datasets~\cite{cencetti_temporal_2021,benson_simplicial_2018}. 
However, work on the modeling of dynamics on temporally evolving hypergraphs is still very limited and, to the best of our knowledge, only exists for the case of discrete dynamics of simplicial contagion \cite{chowdhary_simplicial_2021}.
In this work, we thus extend the literature by considering continuous dynamical processes, in particular general consensus dynamics, taking place on temporal hypergraphs and examine the effects of temporal ordering of multi-body interactions on the overall dynamics.

%% file: 2_model.tex
\section{Consensus dynamics on temporal hypergraphs}
\label{sec:model}

In this section we first recap our nonlinear consensus model~\cite{neuhauser_multibody_2020,neuhauser2021consensus} on hypergraphs and then extend its definition to temporal hypergraphs.
For simplicity, we will concentrate in our exposition on 3-regular hypergraphs, i.e., hypergraphs in which all (group) interactions take place between 3 nodes.
However, our formulation can easily be translated to hypergraphs with hyperedges of different (and possibly mixed) order.

\subsection{Dynamics on static hypergraphs}
We consider dynamics on hypergraphs that emerge from a linear combination of (nonlinear) dynamical interactions mediated by the hyperedges.
Specifically, the dynamics of each node $i$ is governed by the following differential equation:
\begin{equation}\label{eq:general_dynamics}
        \dot{x}_i =   \sum_{j,k} \mathbf{A}_{ijk} f_i^{\{ j,k\} }(x_i,x_j, x_k),
	\end{equation}
where for each hyperedge $(i,j,k)\in E(\Hy)$ the \emph{interaction function} $f_i^{\{ j,k\} }(x_i,x_j, x_k)$ describes the joint influence of nodes $j$ and $k$ on node~$i$.
Specifically, we will be primarily concerned with the 3-way consensus model (3CM) that we introduced in recent work~\cite{neuhauser_multibody_2020}, which follows the scheme of~\cref{eq:general_dynamics} and considers an interaction function is of the form:
\begin{align}
    f_i^{\{j,k\}}(x_i,x_j,x_k) =  \frac{1}{2}\,s \left( \left| x_j-x_k\right|\right) \,\left((x_j-x_i)+(x_k-x_i)\right).
    \label{eqn:our_function}
\end{align}
For each 3-edge $\{i,j,k\}$, the multi-way influence of nodes $j$ and $k$ on node $i$ by the standard linear term $\left((x_j-x_i)+(x_k-x_i)\right)$ 
are modulated by a \emph{scaling function} $s\left(\left|x_j-x_k\right|\right)$ of their state differences.

By choosing the scaling function $s$ in~\cref{eqn:our_function} appropriately, we can create nonlinear dependencies between the node states, such that the influence of nodes $j,k$ on $i$ is modulated by how closely aligned the (opinion) states of node $j$ and node $k$ are.
As we want to be able to model a reinforcement effect for nodes with a similar opinion, a natural choice for the scaling function $s(x)$ is
\begin{align}
	\label{eqn:exponential}
	s \left( \left| x_j-x_k\right|\right)=\exp\left(\kappa \left|x_j-x_k\right|\right),
\end{align}
where the sign of the parameter $\kappa$ determines if the function monotonically decreases or increases. 
We will use this specific form of $s$ in our simulations, but note that other choices of are possible and have been explored \cite{neuhauser_multibody_2020,neuhauser2020opinion,sahasrabuddhe_modelling_2020,neuhauser2021consensus}.

For our work here it will be instrumental to rewrite the dynamics in~\eqref{eq:general_dynamics} with interaction function~\cref{eqn:our_function} in the form:
\begin{align}
\label{eq:ren}
	\dot{x}_i & = \frac{1}{2}\sum_{jk}\mathbf{A}_{ijk} \,s \left( \left| x_j-x_k\right|\right)\,((x_j-x_i)+(x_k-x_i)) \nonumber \\
	          & = \sum_{jk}\mathbf{A}_{ijk}\, s \left( \left| x_j-x_k\right|\right)\, (x_j -x_i)      
	          = \sum_{j}\mathcal{W}_{ij}(x_j-x_i)						
\end{align}
where we have defined the state-dependent interaction matrix $\mathcal{W}$ whose entries $\mathcal{W}_{ij} =\sum_{k}\mathbf{A}_{ijk}s \left( \left| x_j-x_k\right|\right)$ measure the influence of node $j$ on node $i$ (mediated via the 3-way interactions).
We emphasize that~\cref{eq:ren} does \emph{not} describe a linear pairwise system, as the state-dependency of $\mathcal{W}$ renders the system nonlinear, i.e., $\mathcal{W}$ is a function of the state vector $x(t)$, which is dependant on the 3-way interaction topology.
Nonetheless, the above rewriting enables us to write the system in the compact format:
\begin{equation}
    \dot{x} = -\mathcal{L} x,
\end{equation}
where we have defined the (nonlinear) Laplacian $\mathcal L = \mathcal D - \mathcal W$ where $\mathcal D$ is a diagonal matrix with entries $\mathcal D_{ii} = \sum_j \mathcal W_{ij}$.
(To keep the notation compact, we omit the state dependency of these matrices and use caligraphic symbols to denote state-dependent matrices in the following).

Note that, unless the scaling function $s$ is constant, and thus the interaction function~\cref{eqn:our_function} is \emph{linear}, the system cannot be written as a dynamics over fixed graph with pairwise interactions (see~\cite{neuhauser2020opinion,neuhauser2021consensus} for a more detailed discussion). 
However, when the scaling function is constant $s(x)=1$, the dynamics can be reduced to a linear dynamics on an (effective) static weighted graph.
\begin{equation}
	\dot{x}_i = \frac{1}{2} \sum_{jk} \mathbf{A}_{ijk}\left( x_j - x_i + x_k - x_i \right) =-  \sum_j L_{ij} x_j.
\end{equation}
where $L=D-W$ is equivalent to the (state independent) \emph{motif Laplacian} for triangles which was first introduced in \cite{benson_higher-order_2016}, whose entries are defined analogously to $\mathcal L$.
Note that $L$ is simply the standard Laplacian for a graph with adjacency matrix 
\begin{align}
\label{eqn:rescale}
W_{ij}&=  \sum_{k} \mathbf{A}_{ijk},
\end{align}
which corresponds to a rescaled graph, obtained from weighting each interaction between two nodes by the number of jointly incident 3-way hyperedges. 
Therefore, a reduction of the dynamics on a hypergraph to a (pairwise) dynamics on a graph is possible without loss of information for linear dynamics.

\subsection{Dynamics on temporal hypergraphs}
Let us now consider the 3CM model defined on a temporal hypergraph described via a sequence of adjacency tensors $\mathbf{A}^{(1)},\mathbf{A}^{(2)},\ldots$ that each describes the hypergraph topology for a time-period with length $\tau$, respectively.
Using the above discussed (state-dependent) interaction matrices, we can write this temporal 3CM model as
\begin{subequations}\label{time-switched-nonlinear}
\begin{align}
    \dot{x}(t) &= -\mathcal{L}^{(1)}x(t) &0 \leq t \leq \tau, \\
    \dot{x}(t) &= -\mathcal{L}^{(2)}x(t) &\tau \leq t \leq 2\tau, \\
  \vdots & \nonumber\\
  \dot{x}(t) &= -\mathcal{L}^{(r)}x(t) &(r-1)\tau \leq t \leq r\tau,\\
  \text{with } &\quad  x(0) = x_0,
\end{align}
\end{subequations}
where $\mathcal{L}^{(\ell)}_{ij} = \sum_{jk} \mathbf A_{ijk}^{(\ell)}\; s\left( \left| x_k(t)-x_j(t)\right|\right)$ for $(\ell-1) \tau \leq t \leq \ell\tau$. 

Note that the ordering of the temporal hypergraph does not completely determine the specific interaction matrices at time $t$ in~\cref{time-switched-nonlinear}, since the initial condition will influence the specific nonlinear Laplacians $\mathcal{L}^{(i)}$ as well due to the state dependency of the operators at each time-point $t$.
This makes an analytical examination of the general dynamical process difficult. 
We can however gain some analytical insights if we limit ourselves to linear interaction dynamics, in which case we still have a temporal dependency in the interaction topology, and a possible influence of multi-way as opposed to pairwise interactions; however, as the hypergraph dynamics reduce to (effectively) weighted network dynamics in this case~\cref{time-switched-nonlinear} reduces to a time-varying linear system of the form:
\begin{subequations}\label{eq:time_switched_linear}
\begin{align}
    \dot{x}(t) &= -L^{(1)} x(t) &0 \leq t \leq \tau, \\
    \dot{x}(t) &= -L^{(2)}x(t) &\tau \leq t \leq 2\tau, \\
  \vdots & \nonumber\\
  \dot{x}(t) &= -L^{(r)}x(t) &(r-1)\tau \leq t \leq r\tau,\\
  \text{with } &\quad  x(0) = x_0.
\end{align}
\end{subequations}
which can be solved as
 \begin{align}
x(r\tau)=\exp(-\tau L^{(r-1)}) \cdots \exp(-\tau L^{(0)})x_0.
\end{align}

Under some mild assumptions of the connectivity of the underlying system we can guarantee that the above system converges to an average consensus state eventually \cite{olfati-saber_consensus_2007}.

\subsection{Dynamics on random temporal hypergraphs vs. time-aggregated hypergraphs}
Due to the temporal ordering of the interactions, the convergence behaviour of the consensus dynamics on a general temporal hypergraph can be strongly altered even in the linear case.
This makes a comparison of the average behaviour of a consensus dynamics on a hypergraph to that of a corresponding dynamics on a temporal network modelled by a graph cumbersome. Since a general comparison between the temporal and the aggregated, static model is not possible, we restrict ourselves to study the effects of the temporal behaviour of~\cref{eq:time_switched_linear} over one time window of size $\tau$. To eliminate the effects of a specific ordering we assume that the hypergraph in this window is drawn at random from a multiset of $r$ (possibly repeating) Laplacian matrices $\mathbb{L} = \{L^{(1)},\ldots,L^{(r)}\}$. 
A similar setup was studied for dynamics with diffusive pairwise couplings in \cite{masuda_temporal_2013}. 
The expected state of this random temporal hypergraph system after one time-period of length $\tau$ is then given by
\begin{align}\label{eq:random_hypergraph}
    \langle x(\tau) \rangle = \frac{1}{|\mathbb{L}|} \sum_{\ell=1}^r\exp\left(-\tau L^{(\ell)}\right) x(0)= \hat{Z}(\tau)x(0).
\end{align}

To gain insight about the effect the time-varying topology, the expected dynamics on the random temporal hypergraph can be contrasted with the dynamics on the static, time-aggregated topology:
\begin{align}\label{eq:time_aggregated}
    x(t)=-L^*x(t), \quad \text{with} \quad L^*=\frac{1}{|\mathbb{L}|}\sum_{\ell=1}^{r} L^{(\ell)}
\end{align}  
We can now compare the (linear) dynamical processes governed by the time-aggregated hypergraph~\cref{eq:time_aggregated} and the expected outcome of the dynamics~\cref{eq:random_hypergraph} on a random temporal hypergraph after time period $\tau$.
Note that the latter can equivalently be interpreted as the state vector $x(\tau)$ at time $\tau$ of a linear dynamical system  of the form:
\begin{equation}\label{eq:expected_value}
    \dot{x}(t)=-\hat{L}x(t) \quad \text{with} \quad  \hat{L}=\tau^{-1}\ln(\hat{Z}(\tau)),
\end{equation}
where $\hat L$ is an effective Laplacian interaction matrix for one time-period $\tau$. Note that, as said before, this is only the solution for a specific time-point $t=\tau$ for which we compare the aggregated and temporal dynamics. More details on this construction can be found in the appendix in \Cref{apx:Lhat}.

Whenever $L^* \neq \hat{L}$, the expected dynamics on the temporal hypergraph differ from the dynamics on the time-aggregated topology~(cf.~\cite{masuda_temporal_2013} for the case of graphs). 
The reason for this difference is that averaging and integrating the dynamics does not necessarily commute, as the example in \Cref{aggreagetemporal} indicates as well.
The difference between these two dynamics only vanishes in special cases, e.g, when each pair of matrices $L^{(i)}, L^{(j)}$ commutes~\cite{masuda_temporal_2013}. 
Nonetheless, in certain settings, we can derive an analytical relationship between the eigenvalues of $\hat{L}$ and $L^*$ as the following theorem shows.
Though we are primarily concerned with Laplacian dynamics here, we state this result here for general interaction matrices, which extends the special case considered in~\cite{masuda_temporal_2013} for matrices describing sequences of edges.
\begin{theorem}[Eigenvalue relation between time-aggregated and random temporal interactions]
\label{eigenvaluerelation}
Consider a multiset $\mathbb{M}=\{M^{(1)},\ldots,M^{(r)}\}$ of interaction matrices for which all $M \in \mathbb{M}$ fulfill the property $M^2=cM$ for a constant $c$. 
Then we can relate the eigenvalues $\hat{\mu}(\mathbb{M})$  of the effective temporal matrix 
$$\hat{M} = \tau^{-1} \ln\left( \frac{1}{|\mathbb{M}|} \sum_{\ell=1}^r\exp\left(\tau M^{(\ell)}\right)\right)$$
and the eigenvalues $\mu^*(\mathbb{M})$ of the aggregate matrix $$M^* = \frac{1}{|\mathbb{M}|}\left( \sum_{\ell=1}^rM^{(\ell)}\right)$$ by the following relationship: 
\begin{align}
    \hat{\mu}(\mathbb{M})=\tau^{-1}\ln\left(1+\alpha(c,\tau) \mu^*(\mathbb{M})\right)=:f_c(\mu^*,\tau)
\end{align}
with $\alpha(c,\tau) = (\exp(c\tau)-1)/c$.
\end{theorem}

The proof is provided in the appendix (\Cref{apx:proofThm3.1}).
This general result allows us to investigate sequences of interaction matrix for which the property $M^2=cM$ holds. 
In particular, we will see that this is true for the linear consensus dynamics on temporal hypergraphs we consider in the next section.

%% file: 3_Linear_temporal_effects.tex
\section{Linear consensus dynamics on temporal hypergraphs}
\label{sec:linear}

In this section, we investigate how our group-based consensus dynamics behaves for \emph{linear} consensus dynamics acting on a randomized temporal hypergraph as opposed to the corresponding time-aggregated hypergraph and in comparison to a temporal network setting.

\paragraph{Setup}
To elucidate the effect of the temporally changing topology, we consider the expected dynamics on a random temporal hypergraph (\cref{eq:expected_value}) in comparison to the baseline case a time-aggregated topology (\cref{eq:time_aggregated}).
We study these cases for synthetic systems in which groups of $d$ nodes, interact jointly in each time window, with a focus on two scenarios, which are displayed in \Cref{dgroup}.
First, we consider a system with pairwise interactions represented by a simple graph consisting of a clique comprising the $d$ nodes in each time window.
Second, we consider a setup in which the system is described by underlying three-way interactions. Due to the linearity of the dynamics, the three-way system can be reduced to an effective, weighted graph that describes the influence of the nodes onto each other.
Note that this reduction of our dynamical system on the hypergraph to an equivalent dynamics on an effective, weighted graph is possible for each time window of length $\tau$, as discussed in the previous section.
The convergence properties of such a linear dynamics are governed by the spectral properties of the associated linear operator.
This linear operator will correspond to a differently weighted graph for our model, depending on whether we consider the underlying system to have pairwise or three-way (or in general multi-way) interactions.
We can thus compare the convergences rates based on pairwise interactions, to the convergence rates we obtain when considering 3-way interactions.

\begin{figure}
 \centering
 \includegraphics[width=0.4\textwidth]{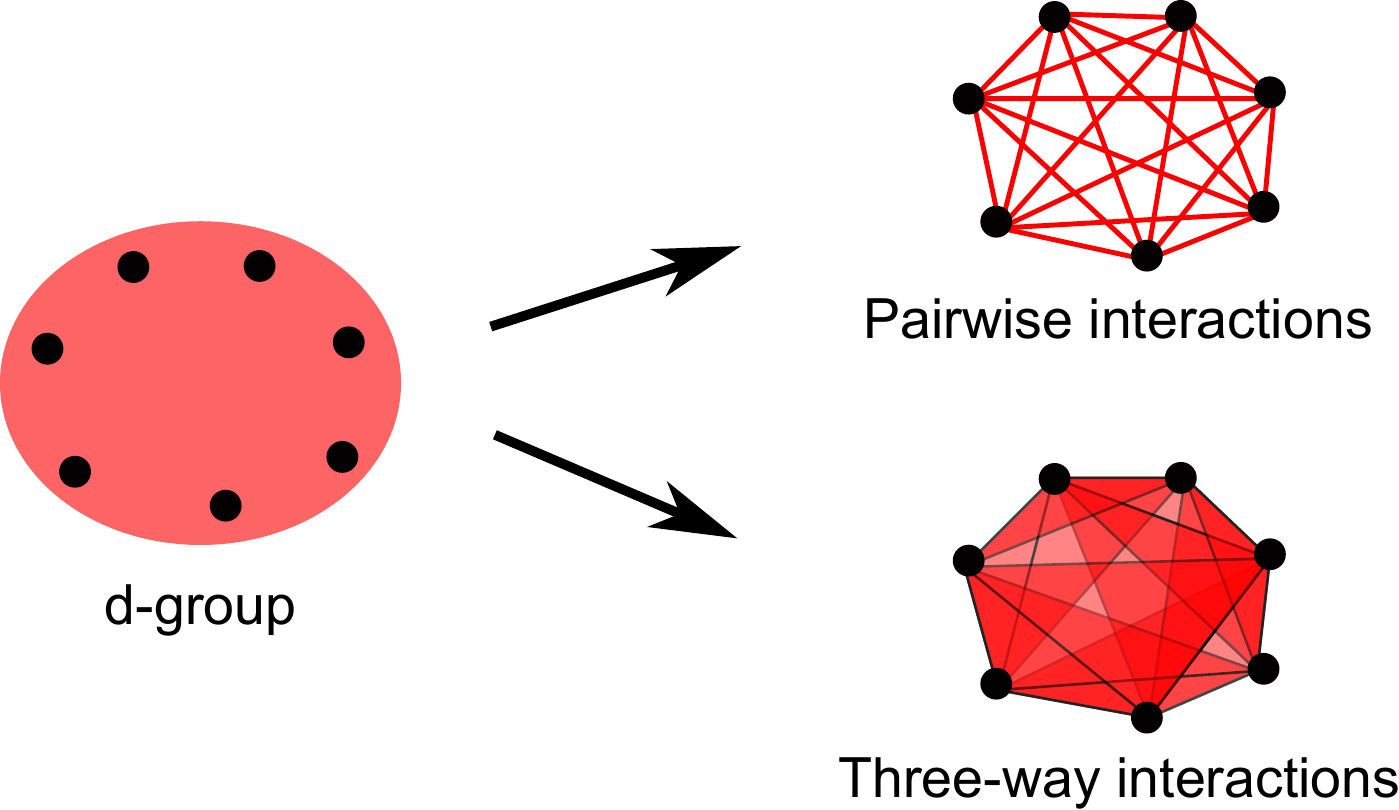}
 \caption[Representation of $d$-groups with pairwise or three-way interactions:]{\textit{Representing d-groups with pairwise or three-way interactions:} We consider groups of $d$ nodes, which interact jointly in each time window. We compare two settings: one in which each $d$-group is a fully connected graph, where all the pairs of distinct nodes are connected, and one in which it is a fully connected three-regular hypergraph, where all the triplets of distinct nodes are connected. Due to the linearity of the dynamics, the three-way system can be reduced to an effective, weighted graph.
}
  \label{dgroup}
 \end{figure}

We assume here that the effective Laplacian $\hat{L}$ (or equivalently of $L^*$) corresponds to a connected graph.
The second smallest eigenvalue of the Laplacian often called the spectral gap, characterises the relaxation time of the dynamical process, e.g. the speed in which the diffusion process reaches a stable state \cite{masuda_random_2017}. This speed increases with an increasing value of the spectral gap.
The eigenvalue relationship between the aggregate and the temporal matrix derived in \Cref{eigenvaluerelation} thus enables us to compare the relaxation time of the temporal and aggregate dynamics.

\paragraph{Pairwise Interactions}
\label{Chapter1Skeleton}
Let  us call a set of $d$ jointly interacting a $d$-group, and consider a multiset $\mathbb{G}_d = \{g_d^{(1)},\ldots,g_d^{(r)}\}$ consisting of $d$-groups.
If we consider a dynamics driven by pairwise interactions, each $d$-group $g_d^{(i)} \in \mathbb{G}_d$ may be suitably represented by a fully connected graph of $d$ nodes. 
A consensus dynamics on such a (motif) graph can then be written as $\dot{x} = - P_{d}^{(i)}x $, where we denote the graph Laplacian for this pairwise interaction model by the matrix $P_d^{(i)} \in \mathbb{R}^{N\times N}$, with entries:
\begin{align}\label{1-skeletonlaplacian}
    P_d^{(i)}= \begin{cases}
    d-1, &i=j \text{ and } i \in g_d\\
    -1 &i \text{ adjacent to } j\\
    0 &\text{else}
\end{cases}
\end{align}

Importantly, using direct computations, it can be shown that $[P_d^{(i)}]^2=dP_d^{(i)}$ for any $g_d^{(i)}$.
Thus, for a set $\mathbb{M}=\mathbb{P}_d$ of interaction matrices $M^{(i)}=-P_d^{(i)}$ for which $[M^{(i)}]^2 = -dM^{(i)}$ holds we can employ~\Cref{eigenvaluerelation} to obtain the following result:
\begin{corollary}
    For a pairwise consensus dynamics on a sequence of $d$-groups, given by $\mathbb{P}_d = \{-P_d^{(1)},\ldots,-P_d^{(r)}\}$, the relationship of the eigenvalues $\mu^*(\mathbb{P}_d)$ of the aggregated interaction matrix and the eigenvalues $\hat{\mu}(\mathbb{P}_d)$ of the effective temporal interaction matrix is
\begin{align}
    \hat{\mu}(\mathbb{P}_d)=\tau^{-1}\ln\left(1+\alpha(-d,\tau) \mu^*(\mathbb{P}_d)\right) .
\end{align}
with $\alpha(c,\tau) = \big(\exp(c\tau)-1\big)/c$.
\label{1-skeletoneignevalues}
\end{corollary}

\paragraph{Three-Way Interactions}
\label{ChapterThreeWay}
As derived \Cref{sec:preliminaries}, a linear dynamical process on a three-way interaction system can be rewritten in terms of a pairwise interaction system using a scaled negative motif Laplacian.

In a $d$-group, assuming underlying three-way interactions instead of pairwise interactions, each interaction between two nodes has$(d-2)$ jointly incident 3-way hyperedges. Therefore, the motif Laplacian for a $d$-group $g_d$ is given as:
\begin{align}
T_d = (d-2)P_d.
\label{eq:motif-Laplacian-relationship}
\end{align}
whereas $P_d$ is the Laplacian capturing the pairwise interactions on the $d$-group $g_d$, which was defined in \cref{1-skeletonlaplacian}. The consensus dynamics for three-way interactions in $g_d$ is then given as $\dot{x}=-T_dx$.

We can now derive a similar relationship as in the pairwise case in \Cref{1-skeletoneignevalues}, this time for a set $\mathbb{M}= \mathbb{T}_d$ of interaction matrices $M^{(i)}=-T_d^{(i)} =-(d-2)P_d^{(i)}$. We can use that $[P_d^{(i)}]^2 = dP_d^{(i)}$ to conclude that $[M^{(i)}]^2 = -(d-2)^2[P_d^{(i)}]^2 = -d(d-2)M^{(i)}$ and employ~\Cref{eigenvaluerelation} to obtain the following result:
\begin{corollary}
For a three-way consensus dynamics on a sequence of $d$-groups, given by the set of interaction matrices $\mathbb{T}_d = \{-T_d^{(1)},\ldots,-T_d^{(r)}\}$, the relationship of the eigenvalues $\mu^*(\mathbb{T}_d )$ of the aggregated interaction matrix and the eigenvalues $\hat{\mu}(\mathbb{T}_d )$ of the temporal interaction matrix is given as
\begin{align}
\hat{\mu}(\mathbb{T}_d )=\tau^{-1}\ln\left(1+\alpha(-d(d-2),\tau) \mu^*(\mathbb{T}_d )\right) 
\end{align}
with $\alpha(c,\tau) = \big(\exp(c\tau)-1\big)/c$.
\label{threewayeignevalues}
\end{corollary}
We can now use these results on the eigenvalues to compare the convergence speed of the temporal and aggregate consensus dynamics. 

\subsection{Temporal Effects of Group Interactions}
\label{subsec:temporaleffects}
We will now apply the results from the previous section to compare pairwise and three-way interactions on $d$-group sequences to investigate how the temporal group dynamics affect the convergence speed. 

Let $\mathbb{T}_d$ be the multiset of three-way and $\mathbb{P}_d$ the multiset of pairwise interaction matrices of a sequence of $d$-groups. As both matrix types represent interactions on $d$-groups and only differ in their type of interaction, the sets are of the same size $\left|\mathbb{T}_d\right|=\left|\mathbb{P}_d\right|$. Additionally making use of the fact that the aggregated matrix is linear in the $d$-group interaction matrices, \cref{eq:motif-Laplacian-relationship} can be directly used to relate the eigenspace of the aggregate matrices of the three-way and pairwise interactions for a sequence of $d$-groups as
\begin{align}
\mu^*(\mathbb{T}_d )=(d-2)\mu^*(\mathbb{P}_d ).
\end{align}
In the pairwise case, as a $d$-group is a fully connected graph on $d$ nodes then, for the $d$-Laplacian it holds that its largest eigenvalue $\lambda_d \leq d$ \cite{zhang_laplacian_2011}. For the aggregate interaction matrix, this value (now sign-flipped) is only possible for a single simplex interaction. Otherwise, we have that
\begin{align}
-d \leq \mu^*(\mathbb{P}_d )< 0
\end{align}
for the aggregated pairwise interactions and
\begin{align}
-d(d-2) \leq \mu^*(\mathbb{T}_d )<0
\end{align}
for the aggregated three-way interactions. 

The following is an extension to the result for edge-sequences from \cite{masuda_temporal_2013} to more general interaction matrices with the property $M^2=cM$ for $c<0$. It implies that diffusion slows down on temporal networks for these types of interaction matrices. The result can be applied for both pairwise and three-way interaction matrices but is not limited to these forms of interactions.
\begin{theorem}[Convergence rates of temporal vs. aggregate dynamics]
\label{weirdlonglemma}
Consider a multiset $\mathbb{M}=\{M^{(1)},\ldots,M^{(r)}\}$ of interaction matrices for which all $M \in \mathbb{M}$ fulfill the property $M^2=cM$ for a constant $c$ with $c < 0$. 
Let the eigenvalues of the aggregated matrix $M^*$  be bounded by $c \leq \mu^* < 0$. Then we have that the temporal dynamics of the sequence are slower than the aggregate dynamics. Moreover, as the time windows $\tau \rightarrow 0$, the temporal dynamics converge against the aggregate dynamics.
\end{theorem}
\begin{proof}
We know that $f_c(\mu^*,\tau)=\tau^{-1}\ln\left(1+\alpha(c,\tau) \mu^*\right)=\hat{\mu}$ with  $\alpha(c,\tau) = (\exp(c\tau)-1)/c$ represents the relationship of the aggregate and temporal dynamics. This function is monotonically decreasing in $\tau$ because $c<0$. We want to look at the limits of this function in $\tau$. Let us only treat $\tau$ as an argument here and set $= \alpha_c(\tau) := = \alpha(c,\tau)$. 

In order to prove the convergence result for $\tau \rightarrow 0$ we use l'Hopitals' rule as we approach $0$ in the denominator. Let us define $h(\tau)=\ln(1+\alpha_c(\tau)\mu^*)$ and $g(\tau)=\tau$. Then $f_c(\mu^*,\tau)=\frac{h(\tau)}{g(\tau)}$. We know that 
\begin{align}
\lim_{\tau \rightarrow 0} \alpha_c(\tau)=\lim_{\tau \rightarrow 0} \frac{\exp(c\tau)-1}{c}=0
\end{align}
and
\begin{align}
\lim_{\tau \rightarrow 0} \alpha'(\tau)=\lim_{\tau \rightarrow 0} \exp(c\tau)=1
\end{align}
With this, we can now compute 
\begin{align}
\lim_{\tau \rightarrow 0} \frac{h'(\tau)}{g'(\tau)} = \lim_{\tau \rightarrow 0} \frac{\alpha'_c(\tau)\mu^*}{1+\alpha_c(\tau)\mu^*}=\mu^*.
\end{align}

When $\tau \rightarrow \infty$, we have to look how $\alpha_c(\tau)$ behaves in the limit. We have that $\lim_{\tau \rightarrow \infty} \alpha_c(\tau)=\lim_{\tau \rightarrow \infty}(\exp(c\tau)-1)/c$. Because $c < 0$, we have that $\lim_{\tau \rightarrow \infty}\alpha_c(\tau)=- 1/c$.We additionally know that $\alpha_c(0)=0$ and the function is monotonically increasing due to the exponential function and $c<0$.  As  $c \leq \mu^* < 0$, we can immediately see that $0 \leq 1+\alpha_c(\tau) \mu^* \leq 1$  for all $\tau \geq 0$, from which follows that
\begin{align}
 \lim_{\tau \rightarrow \infty}f_c(\mu^*,\tau)=\lim_{\tau \rightarrow \infty}\left(\tau^{-1}\ln\left(1+\alpha_c(\tau) \mu^*\right)\right)= 0.
\end{align}
Moreover, as $\alpha_c(\tau)$ is monotonically increasing with increasing $\tau$ and $\mu^* < 0$, the function $f_c(\mu^*,\tau)=\tau^{-1}\ln\left(1+\alpha_c(\tau) \mu^* \right)$ is monotonically decreasing in $\tau$.
\end{proof}

 \Cref{weirdlonglemma} applies to both the pairwise sequence with $c=-d$ and the three-way interactions with $c=-d(d-2)$. This indicates that the convergence of the dynamics on the temporal network is slower than that on the aggregate network, for both pairwise as well as three-way interactions. In particular, the difference between aggregate and temporal dynamics vanishes for small time steps and grows for large time steps, as  $\lim_{\tau \rightarrow 0}f_c(\mu^*,\tau)=\mu^*$ and  $\lim_{\tau \rightarrow \infty}f_c(\mu^*,\tau)=0$.

In the next section, we further investigate how this difference varies from underlying pairwise to three-way interactions to examine a possible interaction effect between higher-order group dynamics and temporal dynamics.
\subsection{Comparison of Temporal Effects of Pairwise and Three-Way Interactions}
\label{Comparison}
In the previous part, we discovered that both pairwise and three-way interactions on temporal sequences of groups have a higher-order effect of slowing down the convergence of the dynamics. We now want to investigate if there is a difference in these effects which would imply an interaction effect between temporal and higher-order group interactions even for linear dynamics.

We compare the two eigenspaces of the temporal matrices for pairwise and three-way interaction matrices with the same aggregated eigenspace $ \mu^*(\mathbb{P}_d )=\mu^*(\mathbb{T}_d )=: \mu^*$:
\begin{align}
  \hat{\mu}(\mathbb{P}_d )&= \tau^{-1}\ln\left(1+ \alpha(-d,\tau)\mu^{*}\right)  \\
 \hat{\mu}(\mathbb{T}_d )& =\tau^{-1}\ln\left(1+ \alpha(-d(d-2),\tau)\mu^{*}\right).
\end{align} We can derive that the function $g(c):=\alpha(c, \tau)=(\exp(c\tau)-1)/c$ is monotonically decreasing in $c$ for all $c<0$, as we have that $g'(c)=(\tau c-1)\exp(\tau c)/c^2 <0$ for all $c <  0$ and the limit $g(c)\nearrow 0$. Combining this and the fact that $\mu^*<0$ (which was motivated in the beginning of \Cref{subsec:temporaleffects}) we can conclude that $\alpha(-d,\tau)\mu^* <\alpha(-d(d-2),\tau)\mu^* $ for $d >3$. As the logarithm is monotonically increasing, it follows that
\begin{align}
 \hat{\mu}(\mathbb{P}_d )  <  \hat{\mu}(\mathbb{T}_d ).
\end{align}
This means that the temporal dynamics with underlying three-way interactions are always slower than the pairwise temporal dynamics for the same aggregated dynamics. This difference is determined by a scaling factor of $k=d-2$ in the nonlinear function $\alpha(-kd,\tau)$ and therefore grows with an increasing dimension of the groups. The scaling factor results from the number of 3-edges that are jointly incident to two nodes in a $d$-group. Therefore, the difference just starts to occur for group sequences with $d>3$. If we look at sequences of 3-edges ($d=3)$, we have  $\hat{\mu}(\mathbb{P}_d )  =  \hat{\mu}(\mathbb{T}_d )$ because each pair of nodes is only part of one 3-edge and therefore the reduction of the three-way system to a weighted graph in each time window does not differ from the pairwise system. As a result, the relaxation times of the temporal dynamics do not differ. 

We can conclude that the difference between temporal and aggregate dynamics of group interactions increases in the case of three-way in comparison to pairwise interactions. This effect depends nonlinearly on the dimension of the group, which reveals an interplay between the two different higher-order effects. We can observe this process for three different dimensions in \Cref{SameAggregate}.

\begin{figure}
 \centering
 \includegraphics[width=0.48\textwidth]{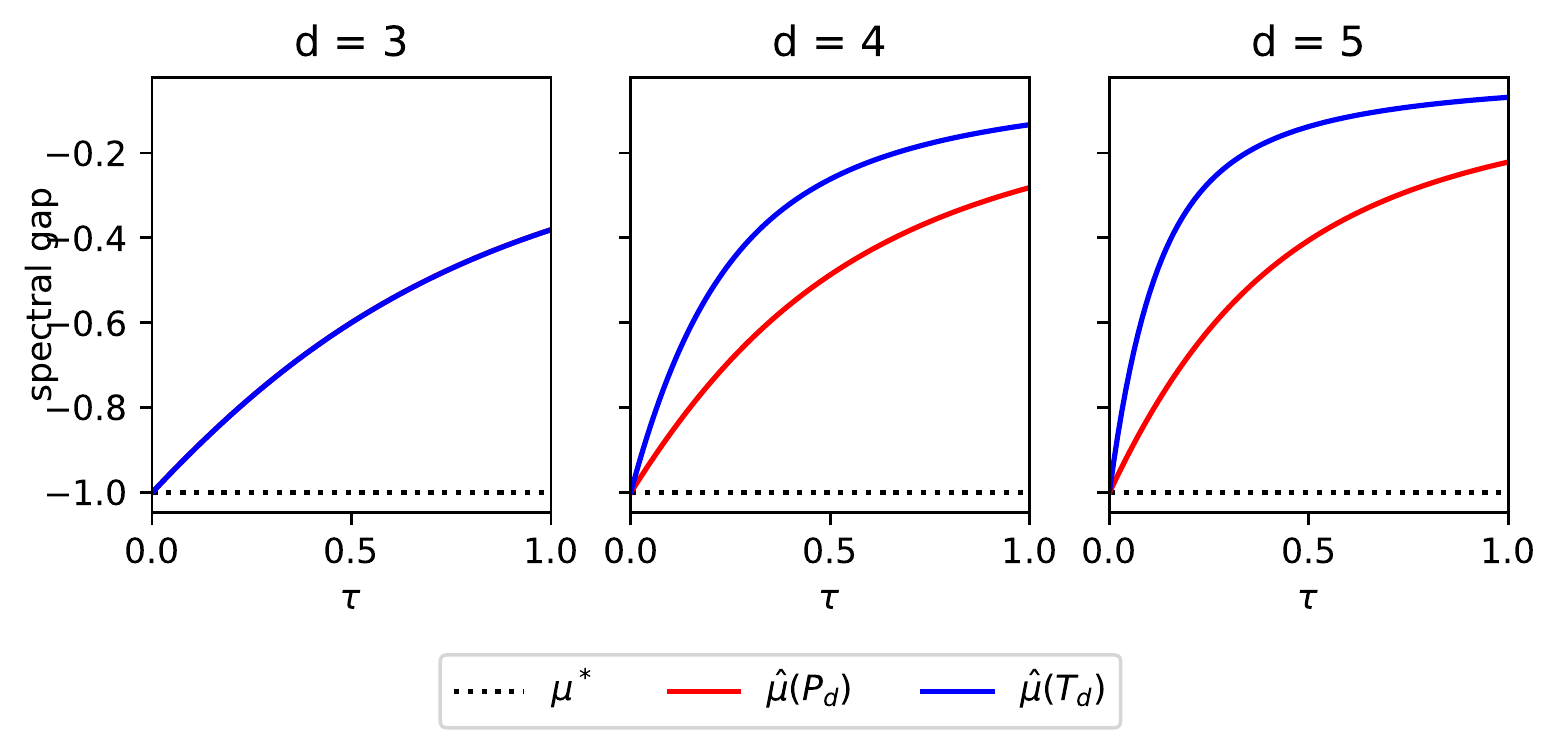}
 \caption[Spectral gap of pairwise and three-way interaction matrices with equal aggregated network]{\textit{Spectral gap of pairwise and three-way interaction matrices with equal aggregated network:} Spectral gap of the temporal (solid line) dynamics of sequences of $d$-groups with both underlying pairwise (red) and three-way interactions (blue) with the same aggregated dynamics. One can observe that both the temporal pairwise and three-way interactions slow down the dynamics compared to the aggregated dynamics, but the effect is stronger for three-way interactions for $d>3$ (for sequences of 3-edges (left), the relaxation times of the temporal dynamics do not differ), scaled by the nonlinear function $\alpha(-kd,\tau)$ with  $k=(d-2)$.  We can observe this in the fact that the slowing effect nonlinearly increases for growing dimension.}
  \label{SameAggregate}
 \end{figure}

%% file: 4_Nonlinear_temporal_effects.tex
\section{Nonlinear consensus on temporal hypergraphs}
\label{sec:nonlinear}
In the previous section, we considered linear interaction dynamics on hypergraphs and provided a theoretical analysis for the impact of (randomly) temporally switching interactions on the convergence properties of an (average) consensus dynamics.
This focus on the convergence speed is justified since for the considered linear average consensus dynamics the final consensus value is not altered irrespective of whether we consider a time-varying or a static hypergraph.

However, this is not the case if we consider a nonlinear consensus dynamics in each time window: in this case, the average opinion within the system is not an invariant of the dynamics, even for a static hypergraph~\cite{neuhauser_multibody_2020}. 
In particular, these opinion shifts are dependant on asymmetries in the system that could, in the static case, occur either on a fully connected system, purely due to the initialisation of the nodes, or due to a clustered hypergraph topology \cite{neuhauser_multibody_2020}. In the following, we always consider initial node states $x(0) \in [0,1]$. Let us denote the initial average of the node states as $\bar{x}(0)$.
In the first case, a consensus formation on a fully connected hypergraph with an initial node state average of $\bar{x}(0) \neq 0.5$ lead to a dominance of the initial majority, in contrast to a scenario with symmetric initialisation of $\bar{x}(0)=0.5$.
In the next section, we will describe the clustered hypergraph case in more detail. 

\subsection{Consensus dynamics on clustered hypergraphs}
\begin{figure*}
\centering
\includegraphics[width=\textwidth]{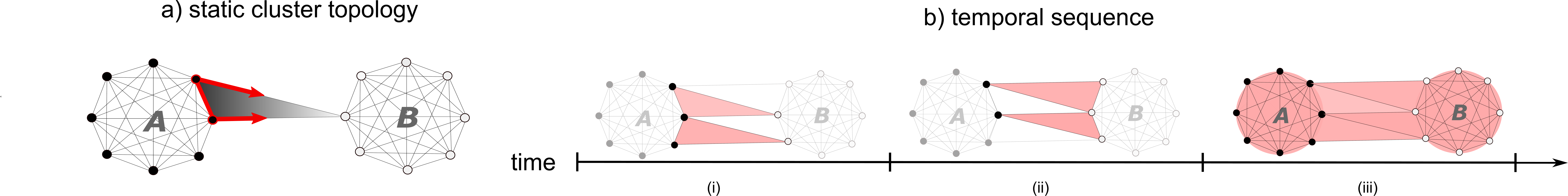}
\caption[Temporal systems of clustered hypergraphs]{\textit{Temporal systems of clustered hypergraphs:} In a) we see the case of a static clustered hypergraph topology of two clusters A and B with initial node states $1$ and $0$ and connected by a 3-edge which is oriented towards cluster B. In this example, group A dominates the consensus process due to the reinforcing dynamics in the connecting hyperedge where A forms the majority. This is indicated by the red arrows. In b) we see a temporal sequence of such a clustered hypergraph, this time with symmetric connection, meaning that an equal number of connecting 3-edges is oriented towards cluster A and B. First, the connecting 3-edges in which group A forms the majority interact in timeslot $(i)$, followed by the other connecting hyperedges in time period $(ii)$. Finally, the whole system is active in $(iii)$. This sequence is referred to as a "First-mover A" scenario, as it captures a setting in which group A has the first-mover advantage of forming the local majority in the connecting subgroups first.}
\label{fig:temporal_sequence}
\end{figure*}

In \cite{neuhauser_multibody_2020} we consider a setting with two clusters which we call $A$ and $B$, see \Cref{fig:temporal_sequence} a).
In this setup there exist two types of hyperedges: If all the nodes of a hyperedge $E_\alpha$ are contained in either cluster $A$ or in cluster $B$, we call $E_{\alpha}$ a \emph{cluster} hyperedge. 
If $E_\alpha$ contains nodes from both clusters, we call $E_{\alpha}$ a \emph{connecting} hyperedge.
A connecting hyperedge is called \emph{oriented} towards cluster $A$, if less of the hyperedge nodes are part of cluster $A$ than of cluster $B$.

We emphasize that all the hypergraph edges are not directed.
Nonetheless, in conjunction with the initial states of the nodes, the distribution of the orientations of the hyperedges can influence the outcome of the consensus formation.
As an extreme example, consider a setup where all nodes in cluster $A$ have initial state  $x_A(0)=1$ and nodes in cluster $B$ initial state $x_B(0)=0$. 
Now consider a connecting hyperedges oriented towards cluster $B$. This scenario is visualised in \Cref{fig:temporal_sequence} a).
Due to the shared opinion in cluster $A$, the two nodes in $A$ reinforce their influence in~\cref{eqn:our_function} on the node in $B$. Vice versa, each node in $A$ is less influenced due to the large state difference between the nodes~\cite{neuhauser_multibody_2020}.
Therefore, we can observe a directional influence between the nodes in the connecting hyperedges, where the opinion of cluster $A$ will be ``propagated'' to cluster $B$ (indicated by the red arrows in \Cref{fig:temporal_sequence} a)).

As shown in \cite{neuhauser_multibody_2020}, depending on the overall orientation of the connecting triangles, this can lead to a result in which the opinion of one cluster ultimately dominates the opinion of the other class:
Specifically, if there is an imbalance in the number of oriented 3-edges towards $A$ or $B$, then one group will have a stronger influence on the final opinion of the consensus process.
In the next part, we will discuss how such consensus dynamics can additionally be influenced by the temporal ordering of the interactions. 

\subsection{First-mover advantages in group discussions}
Let us now consider the above example of a 3-regular hypergraph consisting of two equally sized fully connected clusters with a balanced connection, meaning we have an equal amount of hyperedges oriented towards group $A$ and group $B$. 
We will further assume that the initial states are binary with an initial mean equal to $\bar{x}(0)=0.5$, i.e., a scenario in which even in the nonlinear static case shifts in the final consensus value may only occur due to overlapping connecting hyperedges which can happen due to the random placement. However, if we run many simulations, the mean of $\bar{x}=0.5$ is conserved on average.  
We now investigate, if the temporal ordering of the dynamical interactions can influence the final outcome of the consensus process.

To motivate the above setup, consider a setting in which two parties are debating about a certain issue. 
The members of each party are all in contact with each other, but across party lines, not all people are debating with each other.
Instead, there are certain subgroups of representatives which are meeting to discuss.
These correspond to the connecting 3-edges which involve members of both sides. 
Though overall the oriented edges are balanced and no group has an advantage, in every connecting 3-edge, one party holds the majority.
We thus refer to those 3-edges oriented towards cluster $B$ as A-majority subgroups, and to 3-edges oriented towards cluster $A$ as B-majority subgroups, respectively

In the following, we will compare three different scenarios in which the different subgroups are active at different times.
\begin{enumerate}
\item \textbf{Aggregated system.} This is the baseline in which the hypergraph is static and all interactions are present simultaneously.
\item \textbf{First-mover A}. First, all A-majority subgroups interact for some time, then all B-majority subgroups, then the whole hypergraph interacts.
\item \textbf{First-mover B}. First, all B-majority subgroups interact for some time, then all A-majority subgroups, then the whole hypergraph interacts.
\end{enumerate}
The "First-mover A" scenario is visualised in \Cref{fig:temporal_sequence} b). 

We display the results of the simulations on a hypergraph consisting of two fully connected clusters with $10$ nodes each, a binary initialisation ($x_A(0)=1$, $x_B(0)=0$) and with $20$ connecting 3-edges. 
These 3-edges are randomly placed, such that there is an equal amount oriented towards cluster $A$ and towards cluster $B$.
We then simulate the aggregated system and the two temporal systems described above in the following way: We let each of the majority subgroups interact (as described above) for a fixed time window $\tau$ ($\tau = 200 \cdot 10^{-3}$).
After this, all interactions are active and we let the dynamics run on the whole system until it reaches consensus. For comparison, we show the results of these experiments for both linear and nonlinear consensus dynamics.

For the nonlinear 3CM dynamics, we consider a scaling function of $s(x)=\exp(\kappa x)$ with $\kappa =-100$, so that pairs of similar nodes exert a strong influence on other nodes.  
(the linear dynamics corresponds to $\kappa=0$).
The results are displayed in \Cref{aggregatetemporal} both for the nonlinear reinforcing dynamics (bottom row) and the linear case (top row). 
As we would expect, there is no shift in the final consensus value for the linear dynamics, neither for the static nor the temporal case, as they can be reduced to pairwise dynamics that conserve the average opinion in the system. As discussed in \cref{sec:linear}, the temporal hypergraph setting in the linear case has only an impact on the convergence speed.

The aggregated dynamics show a small shift due to locally different clustering of the randomly placed 3-edges (bottom left). 
This final opinion shift can now be altered by the temporal ordering of the subgroups:
specifically, we can see a first-mover advantage in that the final consensus value the system reaches is closer to the initial opinion of the group which holds the majority in the subgroups which interact first. We can therefore examine a joint effect of temporal ordering and multi-way interactions that can not be found in any of the hypergraph projections.

To understand why this happens consider the ``first-mover A'' scenario as an example.
Observe that the average opinion in the majority-A subgroups is dominated by group A, which leads to a (nonlinear) shift of the average opinion towards the initial states of group A within these hyperedges. 
When the topology switches the initial states for the set of majority-B hyperedges are thus already a lot closer to the states of the nodes of group A, which reduces the possible influence of group B. 
Therefore, for this temporal ordering, the consensus value is mainly influenced by the first-mover group A (bottom center). 
An analogous explanation can be given for the ``first-mover B'' scenario.

\begin{figure*}
    \centering
    \begin{subfigure}[t]{0.33\textwidth}
\centering
\includegraphics[width=\textwidth]{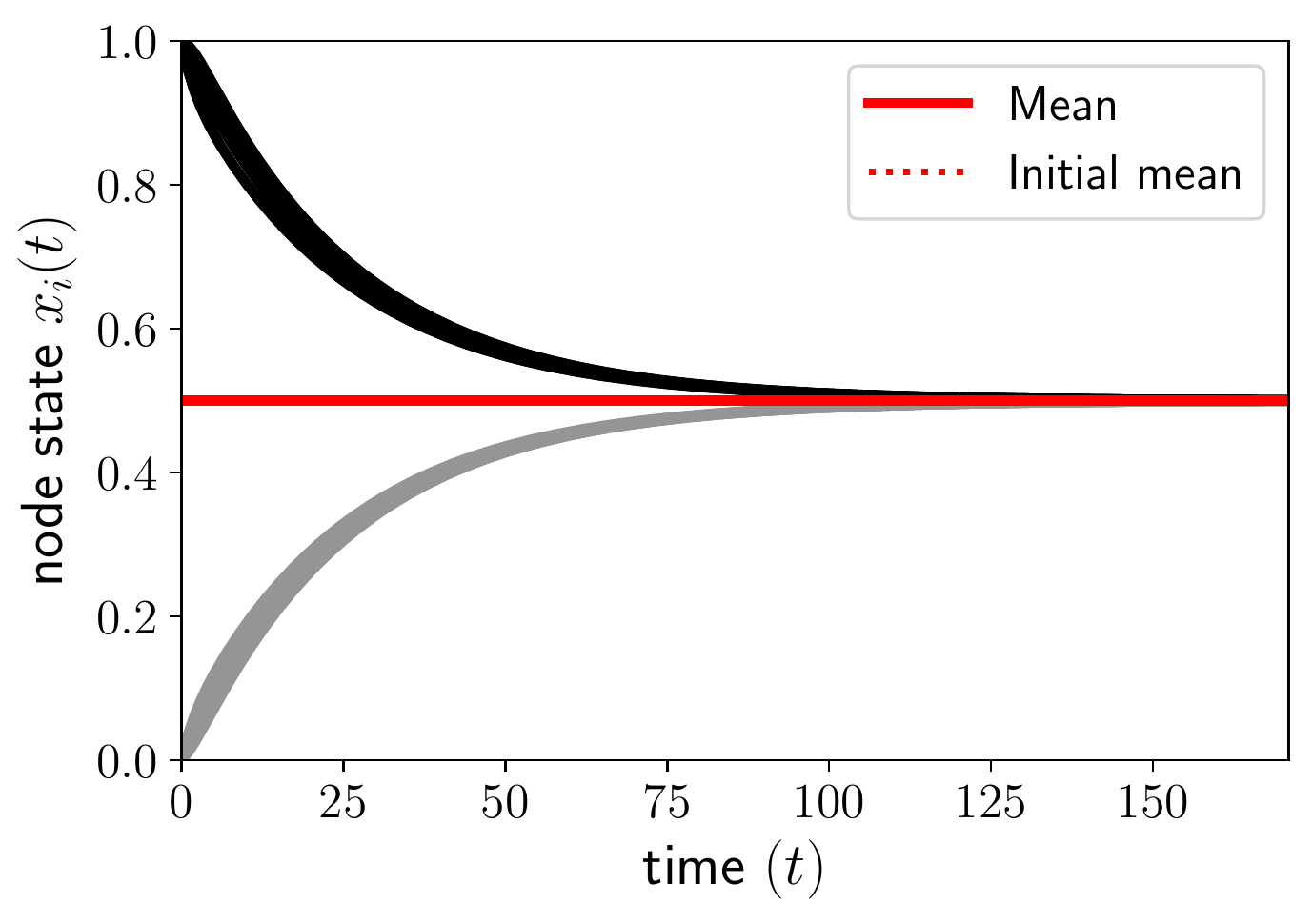}
\caption{linear aggregated}
\end{subfigure}
\begin{subfigure}[t]{0.33\textwidth}
\includegraphics[width=\textwidth]{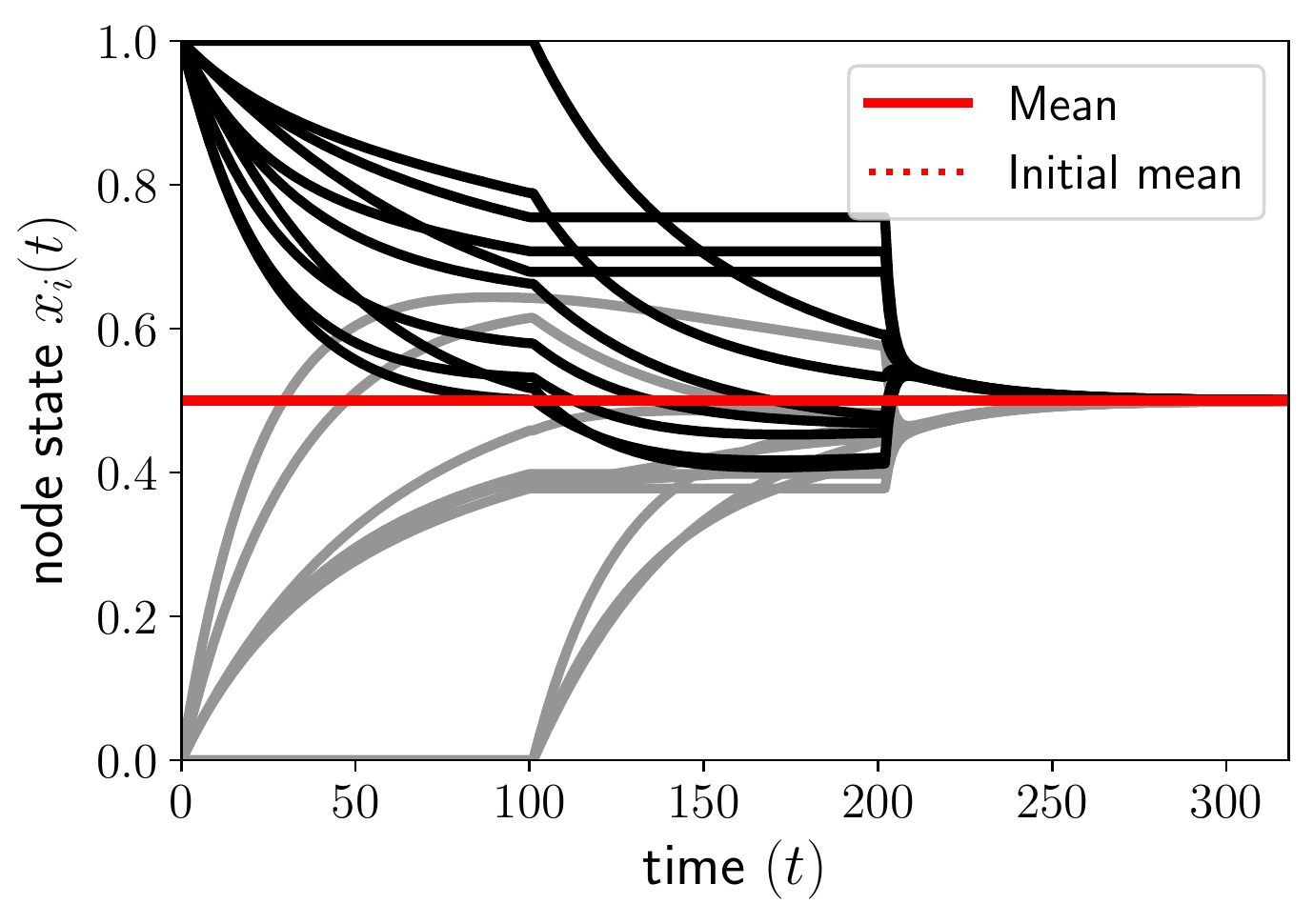}
\caption{\textcolor{blue}{linear first-mover group A}}
\end{subfigure}
    \begin{subfigure}[t]{0.33\textwidth}
\centering
\includegraphics[width=\textwidth]{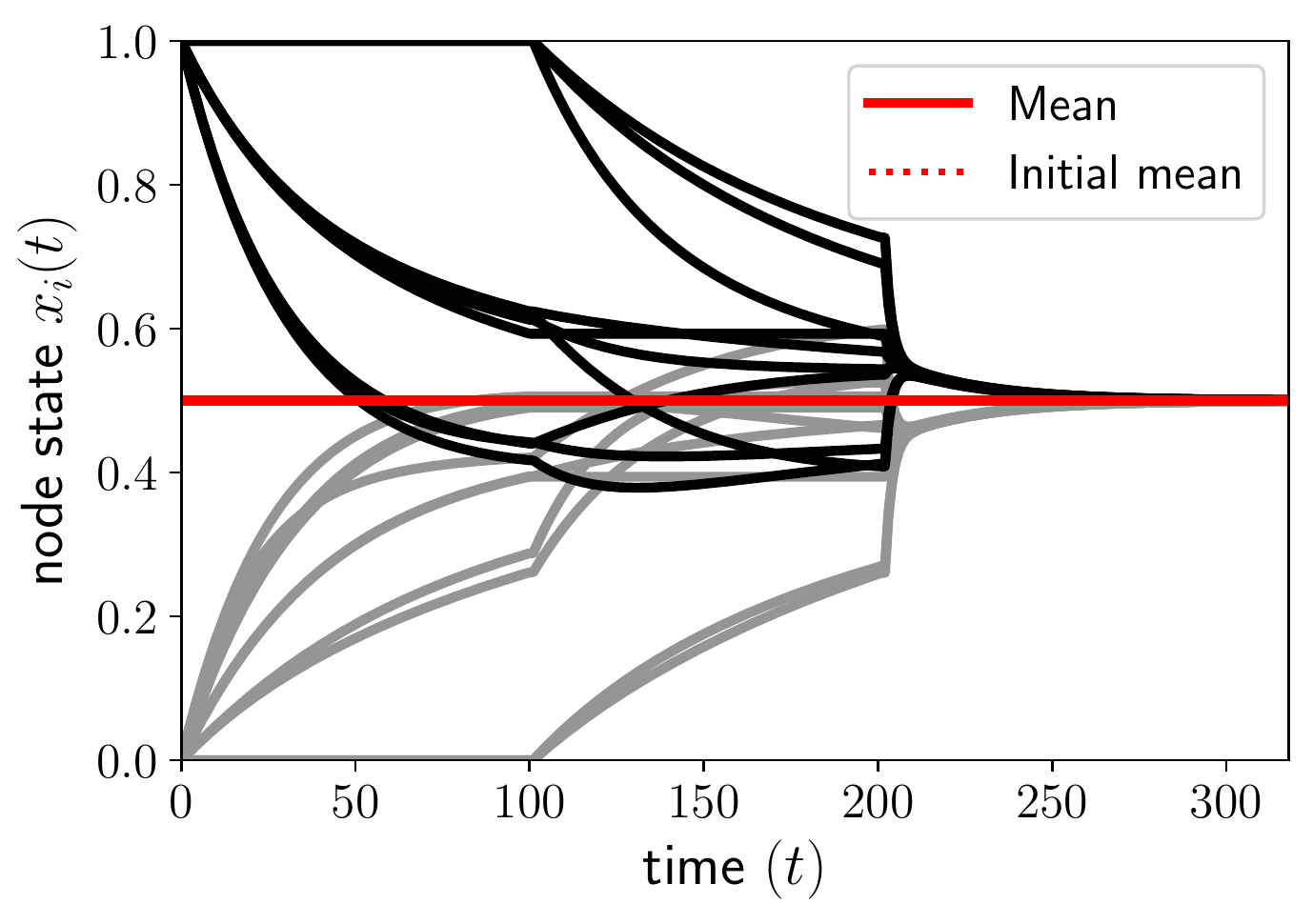}
\caption{\textcolor{ForestGreen}{linear first-mover group B}}
\end{subfigure}
\begin{subfigure}[t]{0.33\textwidth}
\includegraphics[width=\textwidth]{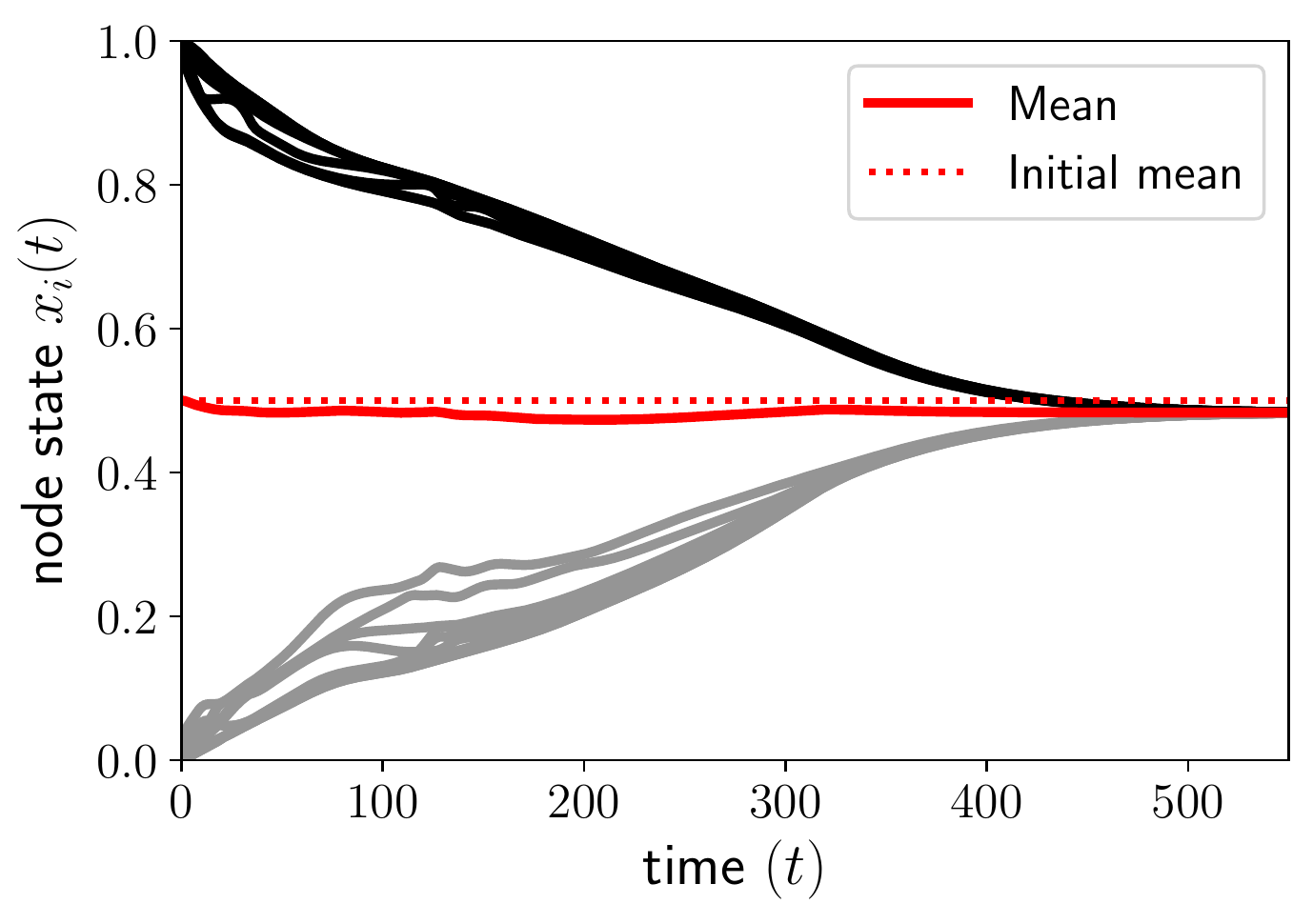}
\caption{nonlinear aggregated}
\end{subfigure}
    \begin{subfigure}[t]{0.33\textwidth}
\centering
\includegraphics[width=\textwidth]{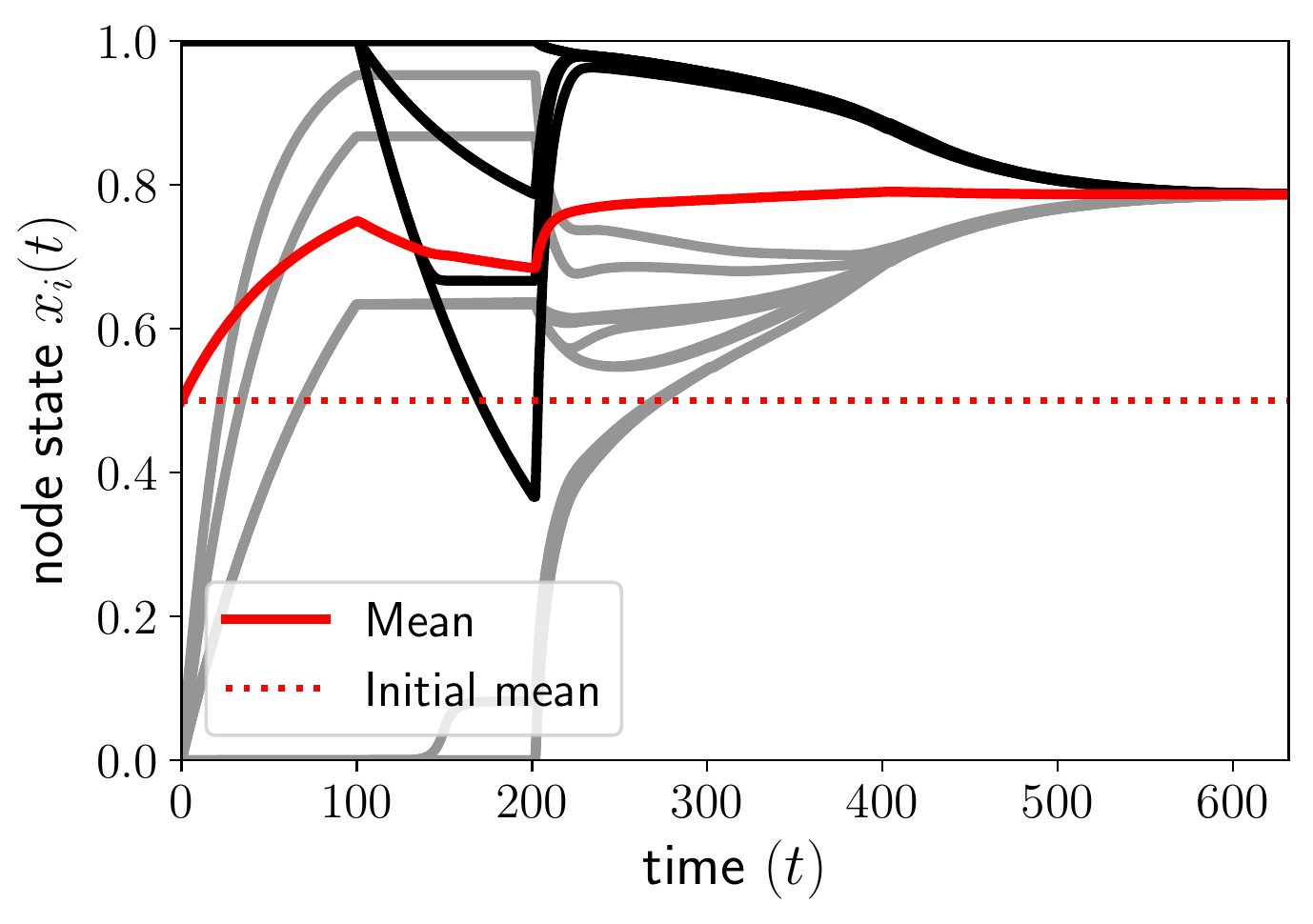}
\caption{\textcolor{blue}{nonlinear first-mover group A}}
\end{subfigure}
\begin{subfigure}[t]{0.33\textwidth}
\includegraphics[width=\textwidth]{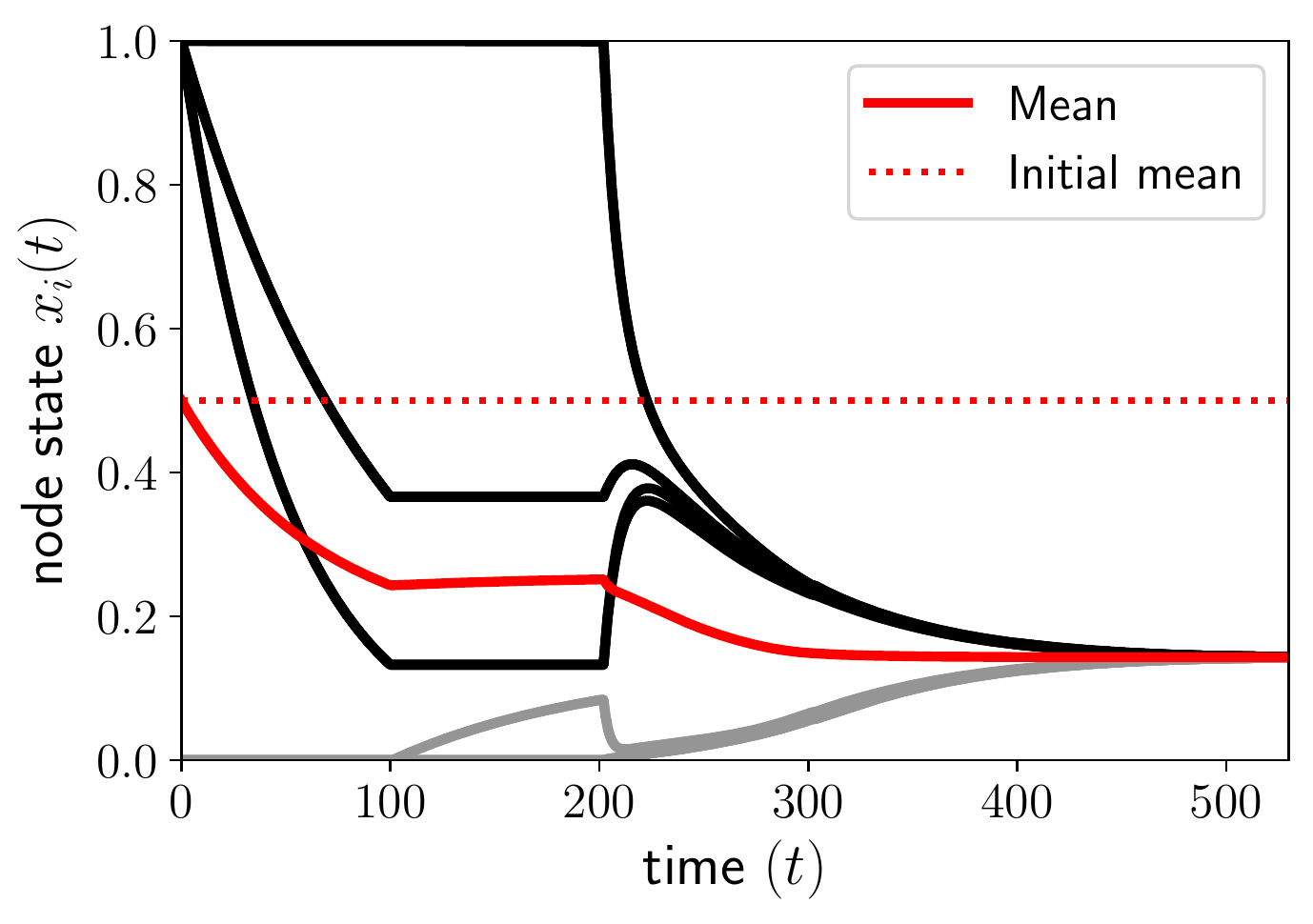}
\caption{\textcolor{ForestGreen}{nonlinear first-mover group B}}
\end{subfigure}
\caption[Nonlinear vs. linear temporal dynamics]{\textit{Nonlinear vs. linear temporal dynamics:} We perform simulations on a hypergraph consisting of two fully connected clusters of 10 nodes, connected by 20 3-edges, from which 10 are oriented towards each cluster, respectively. The nodes of cluster A are initialised with $1$ (black lines), cluster B with $0$ (grey lines). In the linear case, both the aggregated and the temporal systems conserve the initial mean of the node states (top row). The aggregated nonlinear dynamics show a small shift due to locally different clustering of the 3-edges (the connection can hardly be fully symmetric), but this shift is altered by temporal ordering of the subgroups (bottom row). The consensus process is  dominated by the group which has the first-mover advantage. This is the group which forms the majority in the subgroups which are interacting first. The caption colours correspond to the line colours of the respective scenarios in \Cref{temporal_vs_p}.}
\label{aggregatetemporal}%
\end{figure*}

Our simulations confirm that in the case of nonlinear dynamics, there can be a significant difference between temporal and aggregated dynamics on hypergraphs, not only in convergence speed as in the linear case, but also in the final consensus value. 
This is due to the interplay of the multi-way topology and the temporal ordering, which in this case leads to local majorities dominating the consensus process if they interact first. 
In the static case, this dominance was only possible due to an asymmetry in the connection between the two clusters (more connecting hyperedges being oriented towards one group). In clustered hypergraphs with a symmetric connection, the shifts could only occur due to the random placement of the connecting hyperedges, which can sometimes overlap and therefore lead to an asymmetry at random. In the temporal case, the shifts are solely determined by the first-mover groups in this symmetric case.
In the next section, we will investigate the relationship between the cluster topology and the temporal ordering more systematically. 

\subsection{Interaction of cluster topology and temporal effect}
In order to quantitatively analyse the interaction of cluster topology and temporal ordering, we again perform numerical simulations on two fully connected clusters,  each consisting of $10$ nodes, with binary initialisation.
We now connect the clusters with $20$ randomly placed $3$-edges, such that a fraction  $p \in [0,1]$ of $3$-edges are oriented towards cluster $A$ and the rest towards cluster $B$. 
We thus deviate from the symmetric case from before, which corresponds to $p=0.5$. 
We then simulate the first-mover A, first-mover B and aggregated dynamics as defined before on these hypergraphs.

In \Cref{temporal_vs_p}, we show the results of our simulations averaged over $20$ instances in which the initial condition is fixed ($x_A(0)=1, x_B(0)=0$) and the hyperedges are placed at random according to $p$. The error bars display one standard deviation.  
In \Cref{temporal_vs_p} (left), we observe a shift in the final consensus value towards the initial value in cluster $A$ (or cluster $B$, respectively) for the aggregated dynamics, which agrees with the results of the static case in \cite{neuhauser_multibody_2020}. 
The direction of this shift depends on the orientation of the connecting 3-edges, quantified by $p$. For $p=0$, all the connecting 3-edges are oriented towards cluster $B$ and the initial opinion of cluster $A$ thus dominates the dynamics. On the contrary, for  $p=1$ all connecting 3-edges are oriented towards cluster $A$, therefore we observe a maximal influence of the initial value of cluster $B$. 

This qualitative shift is the same for both first-mover scenarios, but the effects of the topology can be reinforced if the first-mover group is the same as the topologically dominating group. 
On the other hand, the effect of the asymmetric topology that leads to an advantage of one group can be overthrown (for $p=0.5)$ or at least diminished ($p=0.2 $,$p=0.8$), if the opposite group has the first-mover advantage. 

This also affects the rate of convergence, as shown in \Cref{temporal_vs_p} (centre). 
Generally, convergence is faster when the initial configuration is very asymmetric. 
Moreover, convergence is significantly faster when topological and temporal dominance of a group align (e.g. group A first-mover effect and $p<0.5$, such that most 3-edges have a group A-majority).

\begin{figure*}
    \centering
    \begin{subfigure}[t]{0.33\textwidth}
\centering
\includegraphics[width=\textwidth]{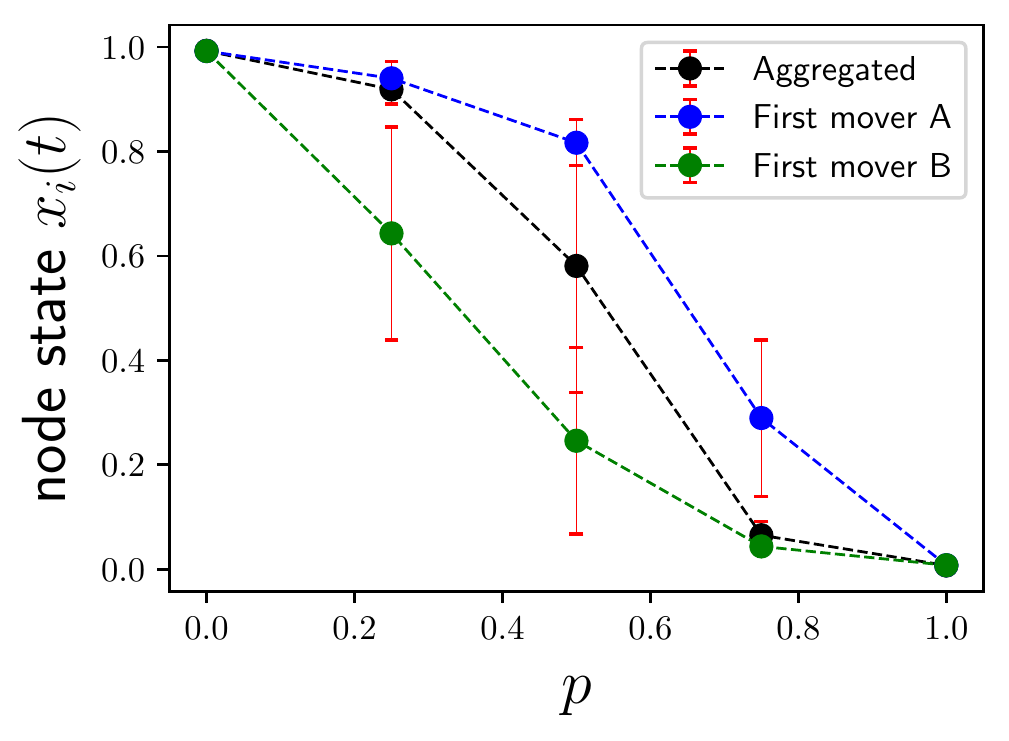}
\caption{Consensus value}
\end{subfigure}
\begin{subfigure}[t]{0.33\textwidth}
\includegraphics[width=\textwidth]{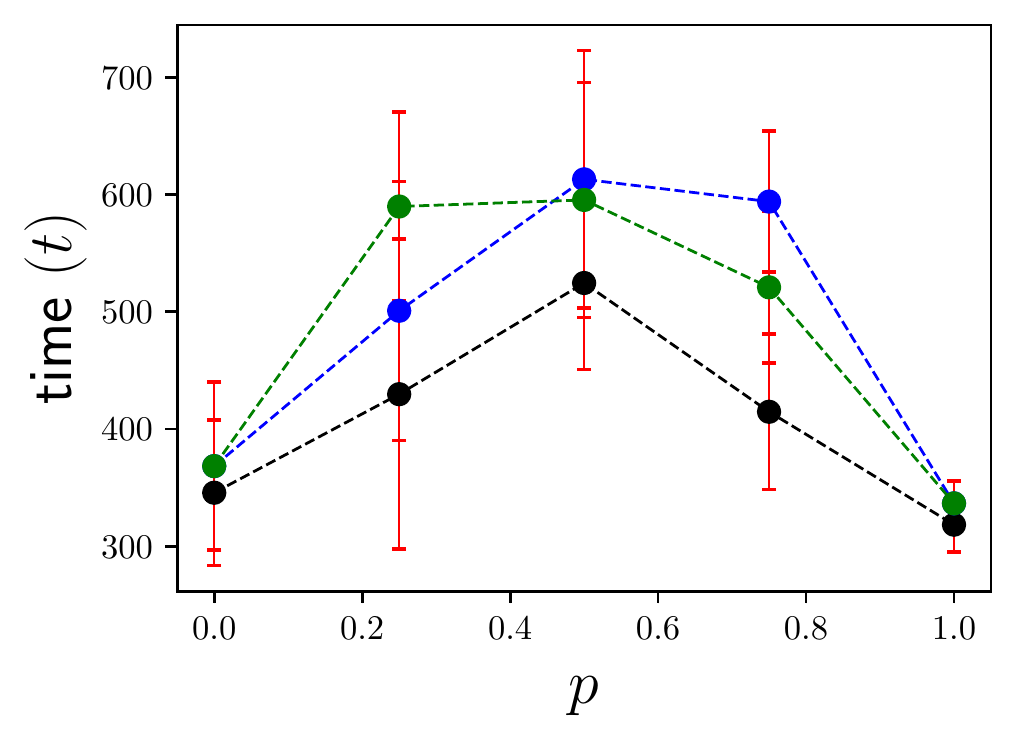}
\caption{Time until consensus}
\end{subfigure}
\begin{subfigure}[t]{0.33\textwidth}
\centering
\includegraphics[width=\textwidth]{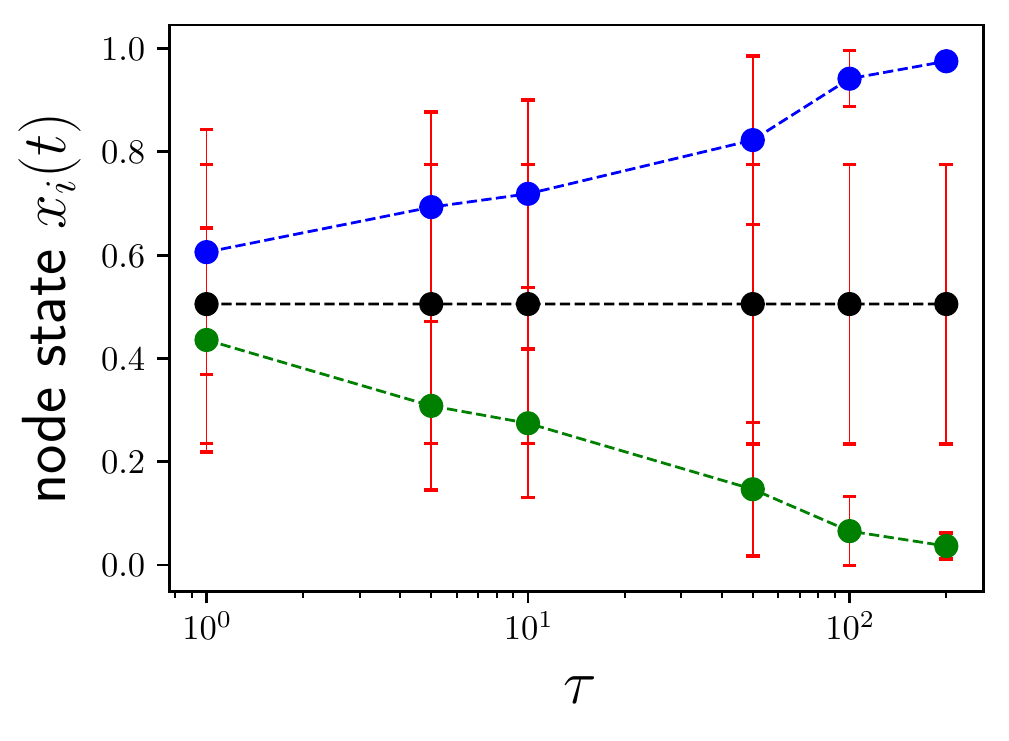}
\caption{Dependence on time scale $\tau$}
\end{subfigure}
\caption[Consensus value, convergence speed and dependence on timescale]{ \textit{Consensus value, convergence speed and dependence on timescale:} We compute the final consensus value, averaged over 10 simulations, where the error bars display one standard deviation (left). As the fraction of 3-edges directed from cluster B to cluster A increases, so does the consensus value towards the initial state in cluster B. This effect of the topology is reinforced, if the first-mover group is the same as the topologically dominating group (e.g. $p<0.5$ and first-mover A blue line) and diminished, if the opposite group has the first-mover advantage (e.g. $p<0.5$ and first-mover B green line). In the symmetric case of $p=0.5$, the shift is solely determined by the temporal ordering, i.e. by which group has the first-mover advantage. The line colours correspond to the caption colours of the respective scenarios in \Cref{aggreagetemporal}.
The rate of convergence is significantly faster when the initial configuration is very asymmetric, that is, extreme values of $p$ (centre). Moreover, the convergence is sped up if the orientation of the 3-edges and the first-mover group align. Additionally, we can see that the convergence value is dependant on the timescale $\tau$ (right). The shorter the timescales of the interactions the more the temporal networks converge to the aggregated dynamics.}
\label{temporal_vs_p}
\end{figure*}

Moreover, as in the case of linear dynamics, we want to know what happens if we let the interaction windows $\tau$ of the temporal snapshots decrease. 
As in the initial simulations, we consider a balanced topology by fixing $p=0.5$. We can see in \Cref{temporal_vs_p} (right), that the temporal dynamics converge to the aggregated dynamics for decreasing $\tau$. This effect is in line with the results of the linear dynamics in Figure \ref{SameAggregate}.

%% file: 5_conclusion.tex
\section{Discussion}
\label{sec:conclusion}
In this work, we investigated the combined effect of two types of higher-order dimensions for networks: their temporality and their multi-way interactions. In particular, we focused on consensus dynamics on temporal hypergraphs. In order to identify the effects of the higher-order representations, we compared the dynamics with appropriate projections that flatten either the temporal or the multi-way component, or both.
Initially, we focused on the linear case.
We extended a key technical insight from~\cite{masuda_temporal_2013} for random temporal graphs to hypergraphs, by combining it with our recent result~\cite{neuhauser_multibody_2020, neuhauser2020opinion,neuhauser2021consensus} that a linear consensus dynamics on hypergraphs can be rewritten as re-scaled pairwise dynamics. 
We then compared the relaxation time of the temporal and aggregated dynamics for underlying linear pairwise and three-way interaction. 
Our analysis reveals that consensus dynamics with 3-way (time-varying) interactions is slowed down in comparison to (time-varying) pairwise interactions.
This slow-down happens \emph{in addition} to the slow-downs that can be observed when considering consensus dynamics on temporal vs. aggregated, static networks.
We can conclude that there is an interesting interplay between multi-way and temporal higher-order models. This interplay is already observable for linear interaction functions, but only has an effect on the convergence speed and not on the final consensus value to which the dynamics converges. 

We then extended our analysis to nonlinear consensus dynamics. 
In particular, we generalised the nonlinear consensus model introduced in~\cite{neuhauser_multibody_2020} to the temporal settings and observed that the consensus value of the temporal system differs from the consensus on the aggregated, static hypergraph. 
More specifically, we can observe a first-mover advantage: groups that are in the local majority in early-active hyperedges have a higher influence on the final consensus value. This mechanism is only due to the temporal ordering of the hyperedges and cannot be observed in the dynamics of the hypergraph projections, neither on the aggregated hypergraph nor for the linear dynamics on the reduced network.

Generally, our results indicate that there is an important interplay between the temporal and multi-way interaction facets of higher-order models. 
In this specific case, both dimensions have individual influences.
However, there is also a joint influence on the overall dynamics of the system, which leads to different results, depending on the nonlinearity of the dynamics. In addition to a difference in convergence speed in the linear case, the nonlinear case can thus lead to a first-mover advantage, describing a situation in which certain groups dominate the final consensus value by the order in which they appear in the temporal interaction sequence. This influences group dynamics which may only be represented by a temporal hypergraph. 
This suggests that when working with higher-order models, the different facets of the model should not be considered completely separately. 

In the future, it would be interesting to analyse how much this interplay can influence real-world systems and how much specific interaction orderings, which can be found in real data, are potentially changing the outcome of opinion dynamics qualitatively. 

%% file: 6_appendix.tex
\section{Appendix}
\subsection{Calculation of the effective matrix of the temporal dynamics}
\label{apx:Lhat}
In the following, we outline how we study the difference between the aggregate and temporal dynamics, which is based on the approach used by Masuda et al. in \cite{masuda_temporal_2013}. When we study linear time-switched systems of the form
\begin{subequations}
\begin{align}
    \dot{x}(t) &= -L^{(1)} x(t) &0 \leq t \leq \tau, \\
    \dot{x}(t) &= -L^{(2)}x(t) &\tau \leq t \leq 2\tau, \\
  \vdots & \nonumber\\
  \dot{x}(t) &= -L^{(r)}x(t) &(r-1)\tau \leq t \leq r\tau,\\
  \text{with } &\quad  x(0) = x_0.
\end{align}
\end{subequations}
these systems always refer to a specific, ordered interaction matrix sequence, in this case to $S=(-L^{(1)}, \cdots, -L^{(r)})$. For this sequence, the system can be solved as
 \begin{align}
x(r\tau)=\exp(-\tau L^{(r-1)}) \cdots \exp(-\tau L^{(0)})x_0 := Z(S,\tau)x_0 .
\end{align}
However, this solution can also be achieved by applying a constant, effective matrix $\hat{L}(S)=(r\tau)^{-1}\ln(Z(S,\tau))$ of the sequence $S$ such that
 \begin{align}
x(r\tau)=(r\tau)^{-1}\ln(Z(S,\tau)) x_0 = \hat{L}(S)x_0.
\end{align}

If we now go beyond specific sequences and assume that the hypergraph dynamics in the time window $\tau$ is drawn at random from a multiset of $r$ (possibly repeating) Laplacian matrices $\mathbb{L} = \{L^{(1)},\ldots,L^{(r)}\}$, this leads to an expected state of this random temporal hypergraph system after one time-period of length $\tau$ of
\begin{align}\label{eq:random2}
    \langle x(\tau) \rangle = \frac{1}{|\mathbb{L}|} \sum_{\ell=1}^r\exp\left(-\tau L^{(\ell)}\right) x(0)= \hat{Z}(\tau)x(0).
\end{align}
Note that \cref{eq:random2} will only give the correct solution for the expected value for a randomly switching system up to the switching time $\tau$. In general, time-switched systems cannot be described by simple eigenvalues when it comes to convergence rates. To work around this limitation we can, as we have calculated above, interpret \cref{eq:random2} as the solution to the linear system
\begin{equation}\label{eq:expected_value2}
    \dot{x}(t)=-\hat{L}x(t) \quad \text{with} \quad  \hat{L}=\tau^{-1}\ln(\hat{Z}(\tau)),
\end{equation}
where $\hat L$ is an effective Laplacian interaction matrix for one time-period $\tau$. The point of \cref{eq:expected_value2} is thus to generate an \emph{effective} system which gives us a convergence rate for a particular time-horizon.
This enables us to consider an average convergence rate for a fixed time-horizon $\tau$ for the time-switched system, which we can use for comparison to the static, time-aggregated topology:
\begin{align}
    x(t)=-L^*x(t), \quad \text{with} \quad L^*=\frac{1}{|\mathbb{L}|}\sum_{\ell=1}^{r} L^{(\ell)}.
\end{align}

\subsection{Proof of Theorem 3.1}
\label{apx:proofThm3.1}
\begin{proof}
Let $M$ be a single interaction matrix that fulfils the property $M^2=cM$. We can conclude that for all $i$, $M^i=c^{i-1}M$. Therefore, the temporal evolution operator of the interaction that lasts time $\tau$ can be written as
\begin{align*}
\exp(\tau M) &= \sum_{i=0}^{\infty}\frac{(\tau M)^i}{i!} 
= I + \sum_{i=1}^{\infty}\frac{(\tau M)^i}{i!} = I +\left( \sum_{i=1}^{\infty}\frac{\tau^{i} c^{i-1}}{i!}\right) M \\ 
             &= I + \frac{1}{c} \left(\sum_{i=1}^{\infty}\frac{(c \tau)^i}{i!}\right) M = I + \frac{1}{c}\left( \sum_{i=0}^{\infty}\frac{(c\tau)^i}{i!}-1 \right)M 
\\&= I+\frac{\exp(c\tau)-1}{c} M=: I+\alpha(c,\tau)M
\end{align*}
Using this, we can write the effective temporal matrix for the multiset $\mathbb{M}$ as
\begin{align}
    \hat{M}&= \tau^{-1}\ln\left(\hat{Z}(\tau)\right) = \tau^{-1}\ln\left(|\mathbb{M}|^{-1}\sum_{\ell=1}^r\exp\left( \tau M^{(\ell)}\right) \right)\\
&=\tau^{-1}\ln\left(I + \alpha(c,\tau) |\mathbb{M}|^{-1} \sum_{\ell=1}^r M^{(\ell)} \right) \\
&=\tau^{-1}\ln\left(I + \alpha(c,\tau) M^*\right) .
\end{align}
Therefore, the two matrices have the same eigenspaces and we can relate the eigenvalues as
\begin{align}
\hat{\mu}(\mathbb{M})=\tau^{-1}\ln\left(1+\alpha(c,\tau) \mu^*(\mathbb{M})\right)=:f_c(\mu^*,\tau).
\end{align}
\end{proof}